%% file: augbaltrees.tex
\documentclass[sigplan,screen]{acmart}
\AtBeginDocument{%
  }

\usepackage{amsthm}
\usepackage{cleveref}
\usepackage{caption}
\usepackage{subcaption}
\usepackage{multirow}
\usepackage{makecell}
\usepackage{hyperref}
\usepackage{thm-restate}
\usepackage{microtype}
\usepackage[font=small,labelfont=bf]{caption}
\usepackage{algorithm,algpseudocode}
\usepackage{graphicx}
\usepackage{color}
\usepackage{xcolor}
\usepackage{soul} 
\usepackage{enumitem} 

\setcopyright{cc}
\setcctype{by}
\acmDOI{10.1145/3774934.3786437}
\acmYear{2026}
\copyrightyear{2026}
\acmISBN{979-8-4007-2310-0/2026/01}
\acmConference[PPoPP '26]{Proceedings of the 31st ACM SIGPLAN Annual Symposium on Principles and Practice of Parallel Programming}{January 31 -- February 4, 2026}{Sydney, NSW, Australia}
\acmBooktitle{Proceedings of the 31st ACM SIGPLAN Annual Symposium on Principles and Practice of Parallel Programming (PPoPP '26), January 31 -- February 4, 2026, Sydney, NSW, Australia}
\received{2025-09-01}
\received[accepted]{2025-11-10}

\renewcommand\footnotetextcopyrightpermission[1]{}

\graphicspath{{./graphics/}}
\algdef{SE}[DOWHILE]{Do}{doWhile}{\algorithmicdo}[1]{\algorithmicwhile\ #1}%

\input{macro}

\begin{document}

\title{Concurrent Balanced Augmented Trees}
\settopmatter{authorsperrow=4}
\author{Evan Wrench}
\orcid{0009-0006-2849-8252}
\affiliation{%
  \institution{University of British Columbia}
  \city{Vancouver}
  \country{Canada}
}
\email{ewrench@student.ubc.ca}

\author{Ajay Singh}
\orcid{0000-0001-6534-8137}
\affiliation{%
  \institution{ICS-FORTH}
  \city{Heraklion}
  \country{Greece}
}
\email{ajay.singh1@uwaterloo.ca}

\author{Younghun Roh}
\orcid{0009-0001-2268-8921}
\affiliation{%
  \institution{Massachusetts Institute of Technology}
  \city{Cambridge}
  \country{USA}
}
\email{yhunroh@mit.edu}

\author{Panagiota Fatourou}
\orcid{0000-0002-6265-6895}
\affiliation{%
  \institution{ICS-FORTH}
  \city{}
  \country{}
}
\affiliation{%
  \institution{University of Crete}
  \city{Heraklion}
  \country{Greece}
}
\email{faturu@ics.forth.gr}

\author{Siddhartha Jayanti}
\orcid{0000-0002-2681-1632}
\affiliation{%
  \institution{Dartmouth College}
  \city{Hanover}
  \country{USA}
}
\email{svj@dartmouth.edu}

\author{Eric Ruppert}
\orcid{0000-0001-5613-8701}
\affiliation{%
  \institution{York University}
  \city{Toronto}
  \country{Canada}
}
\email{ruppert@eecs.yorku.ca}

\author{Yuanhao Wei}
\orcid{0000-0002-5176-0961}
\affiliation{%
  \institution{\mbox{University of British Columbia}}
  \city{Vancouver}
  \country{Canada}
}
\email{yuanhaow@cs.ubc.ca}

\renewcommand{\shortauthors}{E. Wrench, A. Singh, Y. Roh, P. Fatourou, S. Jayanti, E. Ruppert, Y. Wei}


\begin{abstract}
Augmentation makes search trees tremendously more versatile, allowing them to support efficient aggregation queries, order-statistic queries, and range queries in addition to insertion, deletion, and lookup.
In this paper, we present the first lock-free augmented \emph{balanced} search tree supporting generic augmentation functions.
Our algorithmic ideas build upon a recent augmented \emph{unbalanced} search tree presented by Fatourou and Ruppert [DISC, 2024].
We implement both data structures,
solving some memory reclamation challenges in the process, 
and provide an experimental performance analysis of them.
We also present
optimized versions of our balanced tree that use
delegation to achieve better scalability and performance (by more than 2x in \crchanges{most} workloads).
\crchanges{Our experiments show that our augmented balanced tree completes updates 2.2 to 30 times faster than the unbalanced augmented tree, and outperforms unaugmented trees by up to several orders of magnitude on 120 threads.}
\end{abstract}

\ccsdesc[500]{Computing methodologies~Concurrent algorithms}

\keywords{Lock-free, Balanced BST, Augmentation, Ordered Statistics Queries, Range Query, Delegation}

\maketitle

\input{sections/intro}

\input{sections/related}


\input{sections/background}

\input{sections/algorithm_nd_pseudocodes}

\input{sections/algorithm_nd}
\input{sections/delegation}
\input{sections/memory_reclamation}
\input{sections/experimental_results}

\input{sections/conclusion}
\clearpage

\bibliographystyle{ACM-Reference-Format}
\bibliography{augbaltrees}
\clearpage
\appendix

\input{sections/algorithm_delegate}

\input{sections/correctness}

\end{document}

%% file: macro.tex
\newcommand{\punt}[1]{}
\newcommand{\cmnt}[1]{}

\definecolor{xxxcolor}{rgb}{0.8,0,0}

\newcommand{\ang}[1]{\langle #1 \rangle}

\newtheorem*{corollarynonum}{Corollary}
\newtheorem*{theoremnonum}{Theorem}
 \newtheorem{theorem}{Theorem}

\newcounter{history}

\newtheorem{observation}[theorem]{Observation}

\newtheorem{invariant}[theorem]{Invariant}

% ...
\algdef{SE}[FUNCTION]{Function}{EndFunction}[3]{{\op{#1}}\ifthenelse{\equal{#2}{}}{}{(#2)}\ifthenelse{\equal{#3}{}}{}{ : #3}}[1]{\algorithmicend\ \op{#1}}%

\newcommand{\obsref}[1]{Observation~\ref{obs:#1}}

\newcommand{\invref}[1]{Invariant~\ref{inv:#1}}

\newcommand{\Lineref}[1]{Line~\ref{lin:#1}}

\newcommand{\ignore}[1]{}
\newcommand{\myparagraph}[1]{\noindent\textbf{#1}}

\newcommand{\abt}{BAT}

\newcommand{\LComment}[1]{\State $\triangleright$ #1}

\newcommand{\true}{\mbox{\textsf{true}}}
\newcommand{\false}{\mbox{\textsf{false}}}
\newcommand{\nil}{\mbox{\textsf{nil}}}

\newcommand{\CAS}{\op{CAS}}

\newcommand{\f}[1]{\mbox{\textit{#1}}} 
\newcommand{\x}[1]{\mbox{\textit{#1}}} 

\newcommand{\remove}[1]{}  
\newcommand{\myedit}[2]{{\color{#1}{#2}}\normalcolor}
\newcommand{\y}[1]{\myedit{blue}{#1}}

\newcommand{\op}[1]{\mbox{\textsf{#1}}}

\newcommand{\delegateTwoF}{BAT-Del}
\newcommand{\delegateOneF}{BAT-EagerDel}

\newcommand{\PropStatus}{PropStatus}

\newcommand{\crchanges}[1]{{#1}}

\newboolean{showcomments}
\setboolean{showcomments}{false} 

\newcommand{\ajcomment}[1]{%
  \ifthenelse{\boolean{showcomments}}%
    {\textcolor{red}{\scriptsize \textbf{[AjComment:} #1\textbf{\scriptsize ]}}}%
    {}%
}

\newcommand{\hao}[1]{%
  \ifthenelse{\boolean{showcomments}}%
    {\textcolor{green}{\scriptsize \textbf{[HaoComment:} #1\textbf{\scriptsize ]}}}%
    {}%
}

\newcommand{\er}[1]{%
  \ifthenelse{\boolean{showcomments}}%
    {\textcolor{magenta}{\scriptsize \textbf{[EricComment:} #1\textbf{\scriptsize ]}}}%
    {}%
}

\newcommand{\pf}[1]{%
  \ifthenelse{\boolean{showcomments}}%
    {\textcolor{orange}{\scriptsize \textbf{[YoulaComment:} #1\textbf{\scriptsize ]}}}%
    {}%
}

\newcommand{\sjsuggest}[1]{%
  \ifthenelse{\boolean{showcomments}}%
    {{\color{red} #1}}%
    {}%
}

\newcommand{\sj}[1]{%
  \ifthenelse{\boolean{showcomments}}%
    {\textcolor{red}{\scriptsize \textbf{[SiddharthaComment:} #1\textbf{\scriptsize ]}}}%
    {}%
}

\newcommand{\ew}[1]{%
  \ifthenelse{\boolean{showcomments}}%
    {\textcolor{brown}{\scriptsize \textbf{[EvanComment:} #1\textbf{\scriptsize ]}}}%
    {}%
}

%% file: sections/intro.tex
\section{Introduction}


Sets and Key-Value Stores---which support insertions, deletions, and queries---are among the most fundamental and widely used data objects.
Sequentially, two efficient classical data structures are used for these objects: Hash Tables and Balanced Search Trees.
Hash Tables are fast;
but balanced trees are 
 \crchanges{significantly} more versatile, since they preserve ordering of keys and support {\em augmentation} of the nodes with extra information, to enable many more  operations efficiently, including
    {\bf aggregation queries},
    such as {\em size}, {\em sum of values}, {\em maximum}, {\em minimum}, and {\em average};
    {\bf order statistics},
    such as finding the $i$th smallest (or largest) key in the set, or the rank of a given key; and
    {\bf range queries} that
    list or aggregate keys 
    in a given range. 
While augmentation of sequential balanced trees is an indispensable and widely-used technique that is discussed in standard undergraduate algorithms textbooks 
\cite{CLRS4-17,GoodrichTamassia,SMDD19-7,SW11},
the technique has evaded concurrent implementation until now.
We provide the first lock-free {\em augmented} balanced search tree 
with unrestricted augmentation and demonstrate its efficiency empirically.

\subsection{Approach and Challenges}
In the single-process setting, a binary search tree (BST) is often augmented by adding  information to each node
to support additional operations.
For example, in an order-statistic tree,
each node is augmented with a \f{size} field that stores the number of keys in the subtree
rooted at that node.  This facilitates the order-statistic queries mentioned above. 
More generally, augmentation adds \emph{supplementary} fields to each node, whose values can be computed using information in the node and its children.
Augmented search trees are building blocks for
many other data structures, including measure trees \cite{GMW83}, priority search trees \cite{McC85} and link/cut trees \cite{ST83}.

An update to an augmented search tree must often modify the supplementary fields of many
nodes.  For example, in an order-statistic tree, an insertion or deletion
must update the \f{size} field of all ancestors of the inserted or deleted node.
This gives rise to two  key challenges in designing a concurrent augmented tree:
all changes to the tree required by an update must appear to take place
atomically, and nodes close to the root become hot spots of contention since
many operations must modify their supplementary fields.

Recently, Fatourou and Ruppert \cite{FR24} described a scheme
for augmenting concurrent search trees.
In particular, they applied the technique to a
lock-free unbalanced leaf-oriented BST~\cite{EFHR14}.
Their technique stores multiple versions of the supplementary fields. Operations that update the tree propagate information about the update
to each node along the path from the location of the update to the root, step by step.  
To ensure all changes appear atomic, the changes only become 
visible to the operations that use the supplementary fields when 
this propagation reaches the root.
Propagation is done \emph{cooperatively}:  if several processes
try to update a node's supplementary fields, they
need not all succeed, because one update can propagate information about many others.
As a bonus, the augmentation scheme's multiversioning provides simple snapshots of the set of keys stored in the search tree.

Most BST operations take time proportional 
to the tree's height, which can be linear in the number 
of keys  in the tree.
Hence,
\emph{balanced} BSTs, which guarantee the height is logarithmic in the number of keys, are often preferable.
We show how to extend Fatourou and Ruppert's augmentation technique
to get a lock-free augmented \emph{balanced} BST.
This requires 
coping with
rotations (rebalancing operations), which can change the structure of the tree at
any location, whereas the original paper
dealt only with insertions and deletions of leaves.
We apply our extension of the augmentation technique to Brown, Ellen and Ruppert's lock-free
implementation \cite{BER14} of a chromatic BST ~\cite{NS96}, which provides balancing guarantees.
We call our data structure {\em \abt\ (Balanced Augmented Tree)}.
\er{I think we said we would de-emphasize this}
\hao{I removed the discussion of memory reclamation challenges}
We also show how to add memory reclamation to both the augmented unbalanced BST \cite{FR24} and the \abt. 


Using the above ideas, we provide the first C++ implementation of both the augmented unbalanced BST \cite{FR24} and of \abt.
We provide an empirical performance analysis 
of the augmented BSTs. These experiments 
show that our \abt\ scales well
and provides order-statistic
queries that are, in some cases, orders of magnitude faster than previous
concurrent set data structures, which could only handle
them by a brute-force traversal of large portions of \crchanges{a snapshot of} the data structure.

Naturally, updating supplementary fields in augmented BSTs adds overhead to insertions and deletions.
We describe a novel mechanism to significantly reduce this overhead
by having processes delegate the work of propagating
information about updates up the tree to one another.
Roughly speaking, when several updates are trying to 
propagate information along the same path up the tree,
one update propagates information about all of the updates to the root, while the other updates wait for it to complete.
In our experiments, this  mechanism improves \abt's performance by up to a factor of~3. 
It can also be applied to speed up the original augmented (unbalanced) BST of Fatourou and Ruppert~\cite{FR24}.


\subsection{Our Contributions} 
\begin{itemize}[noitemsep, leftmargin=*]
    \item 
    {\bf Lock-Free \abt.}
    The first lock-free balanced augmented search tree supporting generic augmentation function.

    \item
    {\bf Implementation.}
    We implement our algorithm in C++ and provide a lightweight memory reclamation scheme. 

    \item
    {\bf Optimization.}
    We design two delegation schemes that reduce contention between processes that propagate augmenting values along intersecting paths up the tree.
    
    \item
    {\bf Performance.}
    Our experiments show that our \abt\ \crchanges{performs updates} between 2.2 and 30 times faster than the augmented unbalanced tree~\cite{FR24} \crchanges{while maintaining slightly better range query performance} across all our experiments.
    Compared to the fastest unaugmented concurrent tree~\cite{blelloch2024verlib}, BAT is up to several orders of magnitude faster in workloads with order-statistic queries or large range queries.    
\end{itemize}

%% file: sections/related.tex
\section{Related Work}
Three papers on augmenting concurrent search trees appeared in 2024: 
Fatourou and Ruppert (FR) \cite{FR24}, 
Sela and Petrank (SP) \cite{SP24}, and
Kokorin, Yudov, Aksenov and Alistarh (KYAA) \cite{KYAA24}.
FR and SP focus on unbalanced concurrent search trees, whereas KYAA supports balanced search trees.
Our approach extends that of FR (detailed in \Cref{augmentation-basic}), which is lock-free and the most general of the three in terms of the augmentation functions it supports.



SP gave a lock-based augmented tree that supports aggregating functions formed using abelian group operators (i.e., a generalization of augmenting nodes with the sizes of their subtrees).
In contrast, our approach imposes no such restriction on the augmentation.
In their approach, update operations announce themselves with timestamps, and each query must gather information from ongoing updates with smaller timestamps than the query using multiversioning.

KYAA use a wait-free hand-over-hand helping approach in which each node has an associated FIFO queue.  Before accessing a node, an operation must join the corresponding queue and help all operations ahead of it before reading or writing to the node.  Their approach applies to balanced trees and generalizes to handle a similar class of aggregation functions as SP's approach.



Our augmentation scheme, like FR's, has the bonus property of providing simple atomic snapshots of the set of keys in the BST.  
Taking snapshots of shared data structures has received much attention recently \cite{FPR19,NHP22,PBBO12,WBB+21,JJJ24, BCW24, JJ25, JJ25-2, blelloch2024verlib}.  
Naïve but inefficient algorithms for order-statistic queries can use such snapshots.
For example, one can count the keys in a given range by
taking a snapshot and traversing all keys in the range. 
This takes time linear in the number of keys in the range.
Our augmented trees can answer these queries much
more efficiently by traversing just two paths of the BST in time proportional to the tree's height.
In our experimental analysis we compare the performance
of these two approaches to answering such queries,
where the snapshots are provided by the general technique
of Wei et al.~\cite{WBB+21}.
Our \abt\ data structure builds on  
Nurmi and Soisalon-Soininen's chromatic tree \cite{NS96},
which was implemented in a lock-free manner by Brown, Ellen and Ruppert~\cite{BER14}.  

\crchanges{Manohar, Wei and Blelloch~\cite{MWB25} presented two techniques for designing lock-free augmented trees that support low-dimensional k-nearest neighbor queries. Their techniques are optimized for the \emph{shallow-effect augmentation} setting where each update operation typically only changes a few augmented values near the updated leaf rather than changing all of the augmented values up to the root. }





%% file: sections/background.tex

\newcommand{\finalized}{\textsc{Finalized}}
\newcommand{\fail}{\textsc{Fail}}

\section{Background}

Chromatic BSTs \cite{NS96} are a variant of red-black trees \cite{GS78} that separate the steps that balance the tree from the updates that insert or delete nodes,
which makes them more suitable for concurrent implementations.
In \Cref{ber}, we discuss a lock-free implementation of chromatic trees \cite{BER14}, which  our \abt\ data structure builds upon.
\Cref{augmentation-basic} describes Fatourou and Ruppert's augmentation technique  \cite{FR24}, which we  extend so that it can be applied to chromatic trees.

\subsection{Lock-free Chromatic Trees}
\label{ber}

Brown, Ellen and Ruppert \cite{BER14} gave a lock-free implementation of a chromatic BST.
It uses \op{LLX} and \op{SCX} operations, which are an extension of load-link and store-conditional operations.
\op{LLX} and \op{SCX} can be implemented from single-word \op{CAS} instructions \cite{BER13} and provide a simpler
way of synchronizing concurrent updates to a data structure.

\op{LLX} and \op{SCX} operate on a collection of \emph{records},  
each consisting of several words of memory, called the record's \emph{fields}.
Records can be \emph{finalized} to prevent further changes to their fields.
A process $p$ can perform \op{LLX}s on a set $V$ of records and then
do an \op{SCX} that atomically updates one field of a record in~$V$ and finalizes records in a specified subset of~$V$.
The \op{SCX} succeeds \emph{only if} no  \op{SCX} has modified
any record in $V$ between $p$'s \op{LLX} of that record and $p$'s \op{SCX}. 


\setlength{\belowcaptionskip}{-10pt}
\begin{figure}
    \includegraphics[width=7cm]{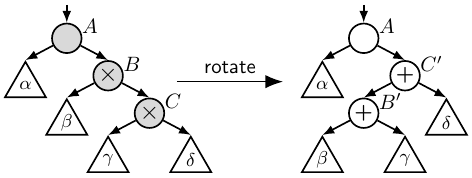}
    \caption{A rotation implemented using \op{LLX} and \op{SCX}.\label{rotate-patch-fig}}
\end{figure}
\begin{figure*}
    \includegraphics[width=0.7\textwidth]{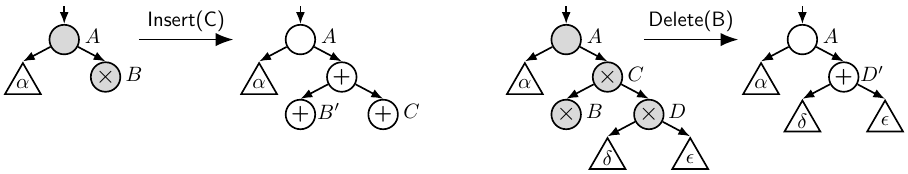}
    \vspace{-5pt}
    \caption{An insertion of $C$ and a deletion of $B$.  In both cases, $B$ is a leaf.  $B'$ and $D'$ are new copies of $B$ and $D$.
\label{insert-delete-fig}}
\end{figure*}
\setlength{\belowcaptionskip}{0pt}

\pf{In the left part of Figure 1, the labels of nodes seem to be on the edges; it would be nice if they are placed closer to the nodes. The same is true for some of the subsequent figures.}

Brown, Ellen and Ruppert \cite{BER14}  described a general technique for implementing
lock-free tree data structures using \op{LLX} and \op{SCX}.
The trees have child pointers, but no parent pointers.
Each tree node is a record.
The record for a node is finalized when the node is removed from the tree.
Starting from the tree's root, an update reads child pointers to arrive
at the location in the tree where it must perform the update.
\crchanges{An update can either be an insertion, a deletion or a rotation.}
Each update to the tree can be thought of as replacing a small group of neighbouring nodes 
in the tree (which we call a \emph{patch}) by a new patch, containing newly created nodes.

For example, \Cref{rotate-patch-fig} depicts \crchanges{one of the 22 possible rotations in the chromatic tree from \cite{BER14}, called RB1}; the nodes of the old patch, marked with $\times$'s, are removed from the tree and the new patch consists of two new nodes, marked with~$+$'s.
The rotation does not change the 
subtrees labelled by Greek letters (some may be empty subtrees).
To perform this rotation, a process first does an \op{LLX} on and reads the shaded nodes
$A, B, C$. 
(If any node of these nodes is \finalized, the rotation is aborted, since some other concurrent process has updated this portion of the tree.)
The process creates new nodes $B'$ and~$C'$, with the same keys as $B$ and~$C$, using 
information returned by the \op{LLX} operations on $B$ and $C$.
Finally, the process uses \op{SCX} to atomically update the child pointer of $A$
to point to~$C'$ and simultaneously finalize nodes $B$ and $C$, which have
been removed from the tree.  If this \op{SCX} succeeds,  no other
updates have modified $A$, $B$ or $C$ after the \op{LLX} operations on them,
ensuring the modification makes the atomic change shown in \Cref{rotate-patch-fig}.  
\Cref{insert-delete-fig} shows other examples of tree updates.

The lock-free chromatic tree~\cite{BER14} is kept balanced by maintaining balance properties that generalize the properties of red-black trees \cite{GS78}. \pf{I think it would be better if this sentence is moved earlier.}
The tree is a \emph{leaf-oriented} BST, meaning keys of the set being represented are stored in the leaves of the tree; internal nodes serve only to direct searches to a leaf.
\crchanges{A few sentinel nodes, each with key $\infty$ are included at the top of} the tree to simplify updates and ensure that the root node never changes.


Our \abt\ operations use the insert and delete operations of \cite{BER14} for the chromatic BST, denoted \op{CTInsert}($k$), \op{CTDelete}($k$).
They use an \op{SCX} to insert or delete a leaf with key $k$, as shown in \Cref{insert-delete-fig}.
If this creates a violation of a chromatic tree balance property,
the update operation is responsible for fixing it before it terminates by applying
rebalancing steps (like the one shown in \Cref{rotate-patch-fig}), again using \op{SCX}. \pf{We may want to say something about the case that the thread executing the update crashes before it applies the rebalancing steps.}
There is at most one balance violation per pending update operation,
and it follows that the height of a tree containing $n$ keys with $c$ pending operations is  $O(\log n+c)$.
\op{CTInsert} returns \true\ if $k$ was not already present, and \op{CTDelete} returns \true\ if it succeeds in deleting $k$; otherwise they return \false.

\subsection{Augmenting Search Trees}
\label{augmentation-basic}

Next, we describe Fatourou and Ruppert's technique for augmenting search trees \cite{FR24}, in particular, their augmentation
of a lock-free (unbalanced) leaf-oriented BST~\cite{EFHR14}.

Each node of the BST has a pointer to a \emph{version} object, which stores a version of the node's supplementary fields.
A node $x$'s version is updated by performing a \emph{refresh} on $x$.
The refresh reads
$x$'s children's versions $\x{v}_\ell$ and $\x{v}_r$, computes the value of $x$'s
supplementary fields, creates a new version object $\x{v}'$ constructed using this
information and finally performs a \op{CAS} to swing $x$'s version pointer to~$\x{v}'$.
In addition to the supplementary fields, $\x{v}'$ 
\crchanges{stores the key of $x$ and} 
pointers to $\x{v}_\ell$ and $\x{v}_r$.  Thus, the version objects
themselves form a BST (called the \emph{version tree}) that mirrors the structure of the original tree (called the \emph{node tree}).  See \Cref{fig-del1}.

After inserting or deleting a leaf of the BST,
an update  must modify the supplementary fields of nodes along the path from the leaf
to the root.  To do so, it performs a refresh (at most) twice at each node along the path.
If a refresh successfully updates the node's version,
then information about the operation has propagated to the node.
If both attempts fail, it is guaranteed that
another process has already propagated information about the operation to the node.
Thus, processes cooperate to carry information about all updates to the root.
(This cooperative propagation \crchanges{technique} originated in a universal construction \cite{ADT95}.)

The search tree stores child pointers but no parent pointers (which would be hard to maintain).
To refresh each node on the leaf-to-root path, an update stores the nodes it traversed to reach the leaf from the root on a thread-local stack.  Then, it can refresh each node as it pops \crchanges{nodes} off the stack to propagate information about the update to the root.

The proof of correctness for the augmented BST defines the \emph{arrival point} of an update operation 
at a node to be the moment when information about the update has been transmitted to the supplementary
fields of the node.  For example, \crchanges{the \op{SCX} step that adds} $C$ into the tree as shown in \Cref{insert-delete-fig}
is the arrival point of the insertion at $C$ and $C$'s parent, because the version objects of those
new nodes are initialized to reflect the insertion.  The first successful CAS by a refresh on some
ancestor $X$ of $C$ that reads $X$'s child $Y$ after the operation has arrived at $Y$ is the arrival
point of the operation at $X$.  The arrival point of the update operation at the root serves as the linearization point of the update.  Much of the proof of correctness is concerned with showing that
operations do arrive at the root before they terminate, and the key invariant that the version tree
rooted at a node's version object accurately reflects all updates that have arrived at that node.

Whenever a refresh updates a node's supplementary fields, it
creates a new version, so the contents of versions are 
immutable.  Thus, when a (read-only) query operation reads the root's version pointer, it
essentially obtains a snapshot of the entire version tree.  
Any query designed for a sequential
augmented BST can therefore be executed on this snapshot without
any adjustment to cope with concurrency.  The query is linearized when it reads the root's version.


%% file: sections/algorithm_nd_pseudocodes.tex
\begin{figure*}
\small
\begin{minipage}{0.49\textwidth}
\begin{algorithmic}[1]
\State \textbf{type} Node \Comment{used to store nodes of chromatic tree}
    \State \hspace{1em} LLX/SCX record containing the following fields:
        \State \hspace{2em} Node *\f{left}, *\f{right} \Comment{pointers to children}
        \State \hspace{2em} Key \f{key} \Comment{tree is sorted based on key field}
        \State \hspace{2em} int \f{weight} \Comment{used for balancing tree}
        \State \hspace{2em} bool \f{finalized} \Comment{node marked as removed}
    \State \hspace{1em} 
    Version* \textit{version} \Comment{pointer to current Version}
\Statex
\State \textbf{type} Version \Comment{stores a node's supplementary fields}
    \State \hspace{1em} Version *\f{left}, *\f{right} \Comment{ pointers to children Versions}
    \State \hspace{1em} Key \f{key} \Comment{key of node for which this is a version}
    \State \hspace{1em} int \f{size} \Comment{number of leaf descendants}
\Statex
\State Node *\x{Root} \Comment{shared pointer to tree root}

\medskip

\Function{Insert}{Key \x{key}}{Boolean}
    \State Boolean \x{result} $\gets$ \op{CTInsert}(\x{key}) with this change: \label{alg:CTInsert}
    \State \hspace*{8mm}Whenever allocating a new Node, apply the
    \Statex \hspace*{12.5mm} Version Initialization Rules to initialize its version.\label{alg:ins-changes}
 	\State \op{Propagate}(\x{key}) \label{lin:ins-propagate}
    \State \Return \x{result}
\EndFunction{Insert}

\medskip

\Function{Delete}{Key \x{key}}{Boolean}
    \State Boolean \x{result} $\gets$ \op{CTDelete}(\x{key}) with this change:
	\State \hspace*{8mm}Whenever allocating a new Node, apply the
    \Statex \hspace*{12.5mm}  Version Initialization Rules to initialize its version.\label{alg:del-changes}
    \State \op{Propagate}(\x{key}) \label{lin:del-propagate}
    \State \Return \x{result}
\EndFunction{Delete}

\medskip

\Function{Find}{Key \x{key}}{Boolean} \Comment{do standard BST search in version tree}
	\State Version* $\x{v}\leftarrow \x{Root}.\x{version}$ \Comment{Start at the root}\label{read-root-find}
	\While{\x{v} has non-\nil\ children} 
        \State $\x{v}\gets(\x{key}<\x{v}.\f{key}\ ?\ \x{v}.\f{left} : \x{v}.\f{right})$
	\EndWhile
	\State \Return{($\x{v}.\f{key} = \x{key}$)}
\EndFunction{Find}

%
\algstore{algoct}
\end{algorithmic}
\end{minipage}
\hfill
\begin{minipage}{0.49\textwidth}
\begin{algorithmic}[1]%
\label{algo:procrefresh}%
\algrestore{algoct}%
\Function{Propagate}{Key \x{key}}{}
    \State Set $\x{refreshed}\leftarrow \{\}$\Comment{stores refreshed nodes} \label{alg:init-refresh}
    \State Stack $\x{stack}$ initialized to contain \x{Root} \Comment{thread-local}
    \Repeat \label{prop-repeat}
        \State Node *$\x{next}\leftarrow \x{stack}.\op{Top}()$
        \Loop \label{alg:propagate-loop} \Comment{go down tree until child is refreshed}
            \State $\x{next}\!\leftarrow\! (\x{key} < \x{next}.\f{key}\, ?\, \x{next}.\f{left}\! :\! \x{next}.\f{right})$
            \State \textbf{exit loop when} $\x{next}\in\x{refreshed}$ or \x{next} is a leaf
            \State \x{stack}.\op{Push}(\x{next}) \label{alg:propagate-endloop}
        \EndLoop 
        \State Node *$\x{top}\leftarrow \x{stack}.\op{Pop}()$
        \If{$\neg \op{Refresh}(\x{top})$}\Comment{if first refresh fails}\label{lin:r1} 
            \State $\op{Refresh}(\x{top})$\Comment{refresh again}\label{lin:r2}  
        \EndIf
        \State $\x{refreshed}\gets \x{refreshed}\cup\{\x{top}\}$\label{alg:prop-update-refresh}
    \Until{$\x{Root}\in\x{refreshed}$}\label{prop-until}
\EndFunction{Propagate}

\medskip

\Function{Refresh}{Node* $x$}{Boolean}
    \State Version* $\x{old} \gets x.version$  \label{lin:readOldVer}

    \Repeat \Comment{get consistent view of left child and its version}\label{lin:Xlbegin}
        \State Node* $x_l \leftarrow x.\f{left}$ \label{lin:readXl}
        \State Version* $\x{v}_l \leftarrow x_l.\f{version}$ \label{lin:readXlV1}
        \If{$\x{v}_l = \nil$} \label{alg:ref-ifl}
            \State \op{Refresh}($x_l$)  \label{alg:rec1}
            \State $\x{v}_l \leftarrow x_l.\f{version}$ \label{lin:readXlV2}
        \EndIf
    \Until{$x_l = x.\f{left}$}\label{lin:Xlend}

    \Repeat \Comment{do the same thing for the right child} \label{lin:Xrbegin}
        \State Node* $x_r \leftarrow x.right$ \label{lin:readXr}
        \State Version* $\x{v}_r \leftarrow x_r.\f{version}$ \label{lin:readXrV1}
        \If{$\x{v}_r = \nil$} \label{alg:ref-ifr}
            \State \op{Refresh}($x_r$)    \label{alg:rec2}
            \State $\x{v}_r \leftarrow x_r.\f{version}$\label{lin:readXrV2}
        \EndIf
    \Until{$x_r = x.\f{right}$}\label{lin:Xrend}

    \State Version* $new \gets$ new Version($\f{key} \gets \x{x.key}, \x{left} \gets \x{v}_l,$  
    \Statex \hfill $\x{right} \gets \x{v}_r, \x{size} \gets \x{v}_l.\x{size} + \x{v}_r.\x{size}$)\label{lin:allocateVersion}
    \State \Return $(\op{CAS}(x.\f{version}, \x{old}, \x{new}) = \x{old})$ \label{lin:refreshCAS}
\EndFunction{Refresh}
\end{algorithmic}
\end{minipage}
\caption{\label{abt-pseudocode}Pseudocode for \abt.  The details of \op{CTInsert} and
\op{CTDelete} on the chromatic tree are provided in~\cite{BER14}.
}
\end{figure*}

%% file: sections/algorithm_nd.tex
\section{Lock-Free Balanced Augmented Tree}
\label{sec:lfact}



We now present \abt.
An update operation first executes the chromatic tree routine
\op{CTInsert} or \op{CTDelete} \cite{BER14} to add or delete a leaf of the chromatic BST
(see Figure \ref{insert-delete-fig}).
\crchanges{These two routines also perform rotations to  eliminate any  balance violations introduced by the update.}
Then, the update must modify the supplementary fields of nodes along
the path from the affected leaf to the root to accurately reflect the update.
This is done using a routine called \op{Propagate}.

To support augmentation, each \abt\ node has a \f{version} field, which 
points to a version object that stores the node's supplementary fields, 
as in the augmented unbalanced BST \cite{FR24} (see \Cref{augmentation-basic}).
A node's 
\f{version} field is not included as part of the \op{LLX}/\op{SCX} 
record that makes up the rest of the node's contents; the \f{version} field can be manipulated directly by \op{CAS} instructions.
This separation ensures that our augmentation does not interfere with
the original chromatic tree 
operations. 

All changes to the chromatic tree of \cite{BER14} are performed by an SCX that replaces one patch of the
BST with a new patch consisting of new nodes and simultaneously finalizes the nodes of the replaced patch.
Whenever a new patch is created for an insertion, deletion or rebalancing step, it uses the following rules to initialize its nodes' versions.
\crchanges{(We use \x{size} as an example augmentation; supplementary fields required by any other augmentation
could be used instead.)}

\begin{definition}[Version Initialization Rules]\label{def:vpr}
Whenever \abt\ creates a new node $x$, its \f{version} field is initialized as follows.
\begin{enumerate}
\item  If $x$ is a non-sentinel leaf (i.e., $x$'s key is not $\infty$), $x.\f{version}$ initially points
to a version with \f{size} $1$.
\item If $x$ is a sentinel leaf, $x.\f{version}$ initially points
to a version with \x{size} $0$.
\item If $x$ is an internal node,  $x.\f{version}$ is initially \nil.
\end{enumerate}
Moreover, the key of every newly-allocated version object is the same as that of the node pointing to it. 
\end{definition}
If a node's \f{version}  is \nil, it indicates that information that should be in the node's supplementary fields is missing. 

We now describe how the \op{Propagate} routine called by an update operation modifies
supplementary fields of nodes starting
from the leaf $\ell$ that was inserted or deleted and moving up to the root.
The supplementary fields are stored in the nodes' versions.
As in \Cref{augmentation-basic}, the basic mechanism for updating a node $x$'s version is a \emph{refresh}
that attempts to install a new version for $x$ that is created using information read from $x$'s children's versions.\pf{Possibly split this sentence to two? }
The goal of \op{Propagate} is to ensure refreshes successfully install new versions at each of a sequence of nodes
$x_1, x_2, \ldots, x_r$, where $x_1$ is the parent of the leaf $\ell$ and $x_r$ is the root so that 
the successful refresh on $x_i$ reads information from $x_i$'s child $x_{i-1}$ \emph{after}
the successful refresh on $x_{i-1}$.
This way, information about the update operation is propagated all the way to the root.

To facilitate this, \op{Propagate} uses a thread-local stack to keep track of the nodes that it should refresh.
The operation first pushes the internal nodes that were visited as it traverses from the root to the leaf $\ell$.
If there were no concurrent modifications to the node tree, \op{Propagate} could simply retrace
its steps, refreshing each node it pops off the stack.
However, concurrent operations on the node tree may have added nodes or replaced nodes that appear in this stack.
Since each refresh transmits information only from children to parent, 
any gap in the chain of refreshed nodes would prevent information about the update from reaching the root. 
Before proceeding to the top node $x$ on the stack, \op{Propagate} checks whether it has already
refreshed $x$'s current child.
If so, it pops $x$ and refreshes it.
Otherwise, it traverses down the tree from $x$ (in the direction towards \crchanges{the update's} key), pushing nodes
onto the stack until it pushes a node $y$ whose child has been refreshed (or is a leaf).
It then pops $y$ and refreshes it,
repeating this process until the root is
refreshed.

\begin{figure*}
\begin{minipage}{0.31\textwidth}
    \includegraphics[width=\textwidth]{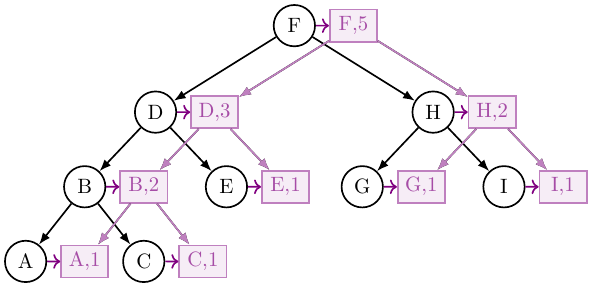}
\subcaption{\abt\ containing keys $A,C,E,G,I$.\label{fig-del1}}
\end{minipage}%
\hfill%
\begin{minipage}{0.34\textwidth}
    \includegraphics[width=\textwidth]{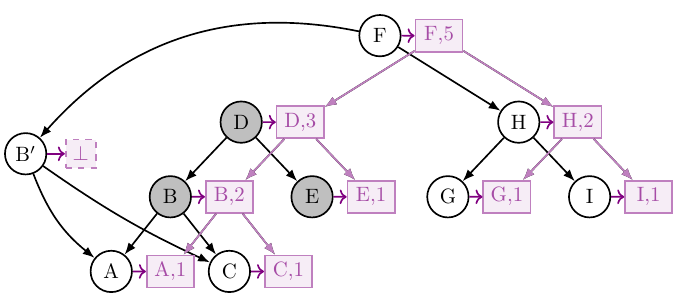}
\subcaption{After \op{CTDelete} removes leaf $E$.\label{fig-del2}}
\end{minipage}%
\hfill
\begin{minipage}{0.34\textwidth}
\includegraphics[width=\textwidth]{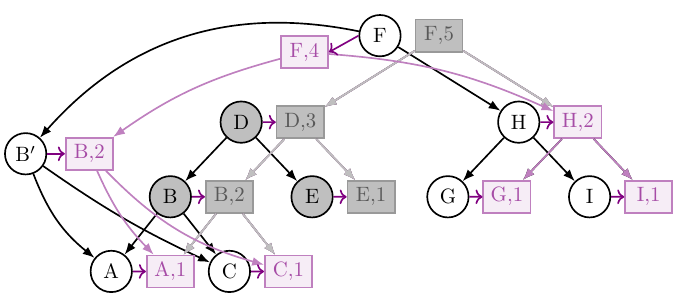}
\subcaption{After \op{Propagate} refreshes $B'$ and $F$.\label{fig-del3}}    
\end{minipage}
    \caption{\label{fig-del}
    Stages of \op{Delete}($E$).  Circles are nodes, rectangles are versions.  Elements shaded gray are unreachable from \crchanges{the root node $F$}.}
\end{figure*}

As in \Cref{augmentation-basic}, 
\op{Propagate} refreshes each node a second time if the first attempted refresh fails.
If the second attempt fails, it is guaranteed that some other successful refresh was performed 
entirely during the interval of the \op{Propagate}'s double refresh.
That other \crchanges{refresh} is guaranteed to have written information into the node that includes
the update the \op{Propagate} is attempting to complete.

All newly created internal nodes are initialized with \nil\ version pointers (\Cref{def:vpr})
to indicate that their supplementary fields have not yet been computed.
If an update \op{op} ever \crchanges{reads} a node $x$'s \nil\ version pointer, 
\op{op} performs a {\em refresh} to fix $x.\f{version}$ 
using information read from $x$'s children's versions.
\crchanges{If the version pointers of $x$'s children are themselves \nil,}
\op{op} recursively tries to fix the children's \x{version} pointers
by reading $x$'s grandchildren. This process goes down the tree recursively until it finds 
nodes whose \x{version} pointers are non-\nil.
(This is guaranteed to happen, since leaf nodes never have \nil\ \f{version} pointers.)
Then, the recursion stops and the version pointers of all nodes 
visited during the recursion are \crchanges{set to non-\nil\ values}.

As in \cite{FR24}, (read-only) query operations simply read the root's version to take a snapshot of the version tree, and run
a \emph{sequential} algorithm on this snapshot, unaffected by concurrent updates.
For example, any order-statistic query in \cite{CLRS4-17} for sequential BSTs can be used verbatim on \abt.

\noindent
\textbf{Description of the Pseudocode.}
Pseudocode for \abt\ is in Figure~\ref{abt-pseudocode}. 
It uses two types of objects.
A Node object represents a tree node and stores a key \x{key},
a \x{weight} (used for rebalancing),  the pointers \x{left} and \x{right} to the
Node's children, and the \x{version} pointer.
The first four fields of a node constitute an \op{LLX}/\op{SCX} record.
A Version object represents one version of the supplementary fields of a node.
It stores a key \x{key}, the \x{size} of the subtree
rooted at the version, and the pointers to the \x{left} and \x{right} children in the version tree.

\crchanges{\abt's
\op{Insert} and \op{Delete} operations simply call \op{CTInsert} and \op{CTDelete} to perform the update on the node tree,
with one change:} the \x{version} field of any newly-allocated node is
initialized using 
rules of Definition~\ref{def:vpr}
(line~\ref{alg:ins-changes} or line~\ref{alg:del-changes}). 
Then, the operation invokes \op{Propagate} (line~\ref{lin:ins-propagate}
or line~\ref{lin:del-propagate}).
Even unsuccessful updates must call \op{Propagate}. 
For example, a process $p$'s \op{Delete}(\x{key}) might fail 
because \x{key} was deleted from the node tree by a concurrent process $q$ 
that has not propagated its deletion to the root, 
so $p$ must ensure $q$'s deletion is propagated before $p$ returns \false. \pf{This is because of the way we assign lin points. Do we want to be more specific here?} 

The loop at lines~\ref{prop-repeat}--\ref{prop-until} of \op{Propagate} ensures that information
about an update operation reaches the root by 
performing a double \op{Refresh} (lines~\ref{lin:r1}--\ref{lin:r2}) on a sequence of nodes, as described above.
It uses \crchanges{the thread-local} \x{stack} to determine which nodes to refresh.
Initially,  \x{stack} is empty. In the first iteration of the outer loop, 
the inner loop at lines~\ref{alg:propagate-loop}--\ref{alg:propagate-endloop} pushes all nodes visited
when searching for key \x{key} onto \x{stack}.
Since \op{CTInsert} and \op{CTDelete} also execute a search for \x{key}, 
we could have those routines store the nodes they visit in \x{stack} before calling \op{Propagate} 
as an optimization.

If \op{Propagate} only refreshes the nodes on the stack, it may skip a new ancestor,
for example one that was rotated onto the search path by 
another update's rebalancing step. 
To ensure that no ancestor is skipped, \op{Propagate} stores the
set of nodes that it has refreshed in its local variable \x{refreshed} (line~\ref{alg:prop-update-refresh}). 
If, at any iteration of the outer loop, the child of the top node $x$ on the \x{stack}
is not in \x{refreshed},
the inner loop at lines~\ref{alg:propagate-loop}--\ref{alg:propagate-endloop}
traverses down the tree from $x$ until reaching a node whose child has been refreshed (or is a leaf).

A \op{Refresh} on node $x$ reads $x$'s left and right child pointers and the children's \x{version} pointers
(lines~\ref{lin:Xlbegin}--\ref{lin:Xlend} and~\ref{lin:Xrbegin}--\ref{lin:Xrend}).
If either child's \x{version} is \nil\ (lines~\ref{alg:ref-ifl} and~\ref{alg:ref-ifr}), 
it recursively refreshes that child (lines~\ref{alg:rec1} and~\ref{alg:rec2}). 
\crchanges{The loops ensure that the node $x_\ell$ or $x_r$
was still a child of $x$ when its \x{version} field was read; if not, the loop rereads $x$'s child and its \x{version} field.}
When all the recursive calls return, \op{Refresh} has all the needed information
to refresh~$x$. Then, it allocates a new version object initializing its 
fields appropriately (line~\ref{lin:allocateVersion}), and attempts to change the \x{version} pointer of 
$x$ to point to this new Version (line~\ref{lin:refreshCAS}). 

The \op{Find} routine does a standard sequential BST search on the version tree rooted at $Root.\f{version}$.
Any other read-only query operation can be done in the same way.

\subsection{Linearizability}

To prove that \abt\ is linearizable, we extend the arguments used for the augmented unbalanced tree \cite{FR24}.
The proof appears in Appendix~\ref{sec:correctness};
we sketch it here.
The goal is to define arrival points of update operations at nodes, ensuring
the tree of versions rooted at a node's version reflects all the operations that have arrived at that node so far. This is formalized as follows.
\begin{invariant}\label{key-inv}
For any node $x$ in the tree that has a non-\nil\ version $v$, 
the version tree rooted at $v$ is a legal augmented BST whose leaves store the set of keys
that would be obtained by executing 
the operations that have arrived at $x$ in the order of their arrival points at $x$ (starting with an empty set).
\end{invariant}
\Cref{key-inv} lets us linearize operations as follows.
An update operation is linearized at its arrival point at the root.
\Cref{key-inv} implies that the version tree rooted at the root's version is a legal augmented
BST that reflects all updates  linearized so far.
Thus, we  linearize each query when it reads the root's \f{version} and gets a snapshot of
the version tree.

Intuitively, an update's arrival point at a node is when information about the operation
is taken into account in the version tree rooted at the node's version. 
The \op{Insert}($C$) shown in \Cref{insert-delete-fig} arrives at the leaf $C$ when the \op{SCX}
adds it to the node tree, since $C$'s \f{version} is initialized to have key $C$ and size 1, according to \Cref{def:vpr}.
The \op{Delete}($B$) in \Cref{insert-delete-fig} arrives at \crchanges{$D'$ when the \op{SCX} changes $A$'s child pointer.  This \op{SCX} is also
the \op{Delete}'s arrival point in
all nodes of $\delta$ on the search path for $B$.}
All of those nodes' versions reflect the absence of key $B$,
since $B$ could not have been counted when constructing their versions.
Arrival points are transferred from a node $x$ to its parent $y$ when a refresh successfully updates $y.\f{version}$:
the refresh's \CAS\ is the arrival point at $y$ of all operations that arrived
at $x$ before the refresh read $x$'s version and that do not already have an arrival point at $y$.

\er{Do we want to mention arrival points of failed updates?}

We must ensure that \Cref{key-inv} is preserved by each modification of the node tree that uses an \op{SCX} to
replace one patch  by another.
As an example, consider the rotation shown in \Cref{rotate-patch-fig}.
This \op{SCX} 
will succeed even if $B$'s \f{version} field has been updated by another process during the time the replacement
patch was being constructed.
\crchanges{It would be} difficult to ensure that the \f{version} field of the new node $C'$ is initialized
to be as up-to-date as the \f{version} field of~$B$.  When the \op{SCX}
changes $A$'s child from $B$ to $C'$, information about operations that had arrived at $B$ (and hence at $A$) might not be included in the version for $C'$.  
If $A$ is then refreshed using information from $C'$, $A$ may lose information about operations that had already previously arrived at $A$,
which would violate \Cref{key-inv}.
Avoiding this bad scenario is the reason we initialize the \f{version} of all new internal nodes (like $C'$) to  \nil, indicating that their supplementary fields must be recomputed when they are needed.
This exempts $C'$ from the requirements of \Cref{key-inv}, which must hold only for nodes with non-\nil \f{version} fields.
On the other hand, it is trivial to initialize the \f{version} field of leaf nodes to accurately reflect
the single key in the leaf.


We must ensure that \Cref{key-inv} is restored when a node's \f{version} pointer is first set to a non-\nil\ value.
Consider the node $C'$ for the rotation shown in \Cref{rotate-patch-fig}.
To avoid violating \Cref{key-inv} as described above,
we must ensure that when $C'.\f{version}$ is changed from \nil\ to a version object $v$,
$v$ reflects all operations that had previously arrived at $A$ from $B$.\pf{Also, those who have arrived at B without yet having arrived at A?}
For this, we use the fact that all such operations must have arrived at $B$ via the roots of
one of the subtrees $\beta,\gamma$ or $\delta$.
When the \nil\ \f{version} pointer of $C'$ is fixed, the recursive refresh routine will ensure that
$B'.\f{version}$ is fixed first, and thus the new version installed at $C'$ will draw upon the latest
information from the roots of $\beta,\gamma$ or $\delta$, ensuring that all operations that had previously
arrived at $B$ will be included in the new version installed at $C'$.\pf{It seems to me that the argument for $\beta$ and $\gamma$ is not exactly the same as that for $\delta$.}
(In the full proof of correctness, we must also consider the possibility that the roots of these subtrees 
have also been replaced by other modifications to the tree after the rotation's \op{SCX}, but a similar argument applies in this case.)
Thus, we define the \op{SCX} that performs the rotation shown in \Cref{rotate-patch-fig}
to be the arrival point at $B'$ of all operations that had arrived at the roots of $\beta$ and $\gamma$,
and the arrival point at $C'$ of all operations that had arrived at the roots of $\beta,\gamma$ and $\delta$.
Even though these operations are not reflected in the (\nil) \f{version} pointers of $B'$ and $C'$, 
we know that when those pointer are changed to non-\nil\ values, all of the operations will be
reflected, restoring \Cref{key-inv}.
It is as if the information about the operations is effectively already in the versions of $B'$ and $C'$ 
as soon as the \op{SCX} performs the rotation, because any operation that reads their \f{version} fields
must first fix them to include that information.

The proof of the analogue of \Cref{key-inv} for the augmented unbalanced BST  \cite{FR24} 
uses another invariant:
if an update operation \op{op} 
with key \x{key} has arrived at a node $x$, 
then \op{op} has also arrived at the child of $x$ on the search path for \x{key}. 
Thus, the set of nodes that \op{op} has arrived at form a suffix of the search path for \x{key} in the BST.
Our definition of arrival points also has this property.
To prove \Cref{key-inv}, we also show that the operation is reflected in the version objects
of all nodes of this suffix that have non-\nil\ version pointers.

%% file: sections/delegation.tex
\section{Reducing Contention via Delegation}
\label{sec:delegation}

A performance bottleneck of \abt, and the original augmented unbalanced BST~\cite{FR24}, is that all updates propagate their changes 
all the way to the root. This causes more cache misses and high contention in the upper levels of the tree. We propose a way to alleviate these drawbacks by having instances of \op{Propagate} delegate their work to concurrent \op{Propagate} instances working along the same path. To ensure linearizability, a \op{Propagate} that delegates its work cannot return until
the \op{Propagate} to which it delegates finishes.
We propose two implementations of this idea.
Detailed pseudocode is in Appendix~\ref{sec:delegate-alg}.


In our first implementation, called \textbf{\delegateTwoF{}}, an instance $P$ of \op{Propagate} delegates its work \crchanges{if} failing \emph{both} of its attempts to refresh a node $x$ at lines~\ref{lin:r1} and~\ref{lin:r2} of Figure~\ref{abt-pseudocode}.
This means some other successful \op{Refresh} on $x$ occurs between the start of $P$'s first failed \op{Refresh} and the end of $P$'s second failed \op{Refresh}.
Let $R_{\ell}$ be the last such successful \op{Refresh} and $P_{\ell}$ be the \op{Propagate} that called~$R_{\ell}$.
After $P$'s second failed \op{Refresh}, 
if $x$ is not finalized (i.e., it is still in the tree),  $P$ delegates the rest of its work to $P_{\ell}$ by 
waiting for $P_{\ell}$ to complete before returning.
\crchanges{We say that the delegation happens at $R_{\ell}$'s successful CAS.}
\crchanges{When $P_{\ell}$ completes, all of the operations $P$ was attempting to propagate will have arrived at the root.}

We next describe how $P$ synchronizes with  $P_{\ell}$.
Each call to \op{Propagate} \crchanges{creates a} \PropStatus\ object, which stores a boolean value \f{done} indicating whether or not the \op{Propagate} has finished and a pointer to another \PropStatus\ if the \op{Propagate} has delegated its work (if not, this pointer is \nil).
Each version stores a pointer to the \PropStatus\ of the \op{Propagate} that created it.
When a \op{Propagate} $P$ fails the \op{CAS} of its second \op{Refresh} on a node, the \op{CAS} returns the version written by the last successful \op{Refresh} $R_\ell$, and $P$ can delegate to that version's \PropStatus, which belongs to $P_\ell$.
$P$ waits for $P_\ell$ to finish by spinning on the \x{done} field in the \PropStatus\ object.
There may be a chain of delegations, so to avoid waiting for the \x{done} flag to propagate down the chain, $P$ can find the head of the chain using the pointers in the \PropStatus\ objects, and directly wait on the head \crchanges{of the chain}.


For correctness, 
we distinguish between \emph{top-level} \op{Refresh}es called by \op{Propagate} and \emph{recursive} \op{Refresh}es called by \op{Refresh} to fix \nil\ version pointers.
We need to ensure that a top-level \op{Refresh} cannot fail (and thus delegate) due to a recursive \op{Refresh}.
We do this by making the \op{CAS} in recursive \op{Refresh}es only change version pointers from \nil\ to non-\nil\ and the \op{CAS} in top-level \op{Refresh}es change version pointers only from non-\nil\ to non-\nil.
\crchanges{We accomplish this by creating two separate refresh functions that only differ in their first few steps. A recursive \op{Refresh} begins by reading the node's version pointer and returning if it is non-\nil{}. A top-level \op{Refresh} begins by reading the node's version pointer and, if it is \nil{}, calling a recursive refresh and rereading the node's version pointer (which is now guaranteed to no longer be \nil). The remaining steps for both versions of \op{Refresh} are the same as in Figure~\ref{abt-pseudocode}.}
Delegating due to a recursive \op{Refresh} is dangerous because a \op{Propagate} may perform recursive \op{Refresh}es on nodes outside of its search path. 
For example, when a new patch is installed, a \op{Propagate} might recursively \op{Refresh} every node in the new patch, even though many are on different search paths.

\crchanges{Next, we sketch why the delegation scheme is correct.}
In general, it is safe for a \op{Propagate} $P_1$ to delegate to another \op{Propagate} $P_2$ at node $x$ if \crchanges{at the time of delegation} (1) $x$ is reachable from the root, (2) $x$ is on the search paths of both $P_1$ and $P_2$, (3) the \op{Refresh} of $P_2$ sees the set $A$ of all arrival points that $P_1$ was attempting to propagate to $x$, and (4) none of the nodes above $x$ are in $P_2$'s refreshed set. Properties (1) and (2) are important for arguing that $P_2$ will perform a sequence of \op{Refresh}es from $x$ (or from a new node replacing $x$, which inherits all the arrival points in $A$) to the root.
This sequence of refreshes
has the same effect as continuing $P_1$ because properties (3) and (4) ensure they will bring the arrival points in $A$ to the root. 

\crchanges{Together, these properties can be used to prove the following invariant by induction over the execution history.
Consider a \op{Propagate} $P$ called by an update operation $U$ with key $k$.
Let $x_C$ be the top-most node (the node closest to the root) on the search path for $k$ that $U$ has arrived at.
The invariant is that at every configuration $C$ during $P$, if $U$ has not arrived at the root, then the \op{Propagate} $Q$ at the end of $P$'s delegation chain is at or below $x_C$, meaning that $x_C$ is on $Q$'s search path and no node between the root (inclusive) and $x_C$ (exclusive) is in $Q$'s refreshed set.
This invariant ensures that $U$ arrives at the root before the root is added to $Q$'s refreshed set.
Since $P$ waits for $Q$, and $Q$ only returns when the root is added to $Q$'s refreshed set, $U$ must arrive at the root before $P$ returns.
Thus, \delegateTwoF{} preserves the correctness of the original algorithm.
}

In our experiments, \delegateTwoF{} improves performance by more than a factor of two in update-heavy workloads. 
We develop a more optimized version called \textbf{\delegateOneF{}}, which further increases the frequency of delegation.
\delegateOneF{} uses the same delegation mechanism, but \crchanges{changes the \op{Propagate} function so that it is safe to delegate} if just \emph{one} \op{Refresh} $R$ fails.
\crchanges{Let $R_\ell$ be the \op{Refresh} that caused $R$ to fail. 
When $R$ delegates to $R_\ell$, $R_\ell$ might not have seen all the arrival points that $R$ was attempting to propagate to $x$, violating property (3).}
To restore the property, we modify the \op{Propagate} function so that it keeps calling \op{Refresh} on the same node until either (1) the \op{Refresh} fails and the node is not finalized (in which case it delegates to the \op{Propagate} that caused this \op{Refresh} to fail), 
or (2) the \op{Refresh} succeeds and the version pointers in the node's children are the same before and after the \op{Refresh}.
Suppose a \op{Propagate} $P_1$ delegates to a \op{Propagate} $P_2$ at node $x$.
Condition (2) ensures that $P_2$ continues to call \op{Refresh} $x$ until it sees that the version pointers in $x$'s children do not change during the execution of such a \op{Refresh}, at which point $P_2$ is guaranteed to have seen all of the arrival points $P_1$ was attempting to propagate to $x$, thus restoring property (3).
The correctness of \delegateOneF{} uses similar arguments as those for \delegateTwoF{}.

Both delegation techniques described above are blocking:  if a thread that has undertaken the work of other threads stalls, it prevents the other threads from completing.
We can make both \delegateTwoF{} and \delegateOneF{} non-blocking by adding a timeout, after which the waiting process resumes its propagation to the root itself.




%% file: sections/memory_reclamation.tex
\section{Memory Reclamation}
\label{sec:memory-reclamation}

In this section, we describe how to apply Epoch-Based Reclamation (EBR)~\cite{fraser2004practical} to free the three types of shared objects used by BAT: nodes, versions, and, when using delegation, \PropStatus\ objects.
EBR tracks the beginning and end of each \emph{high-level} operation (e.g. \op{Insert}, \op{Delete}, \op{RangeQuery}).
It provides a \op{retire} function, which takes as input an object to be freed and delays freeing that object until all high-level operations active during the retire have completed.

Applying EBR (or any other memory reclamation scheme) to BAT poses several novel challenges.
Traditionally, EBR retires objects when they are removed from the data structure.  This is easy to track for tree nodes  since they can be reached by only a single path from the root.  
For example, nodes marked by $\times$ in \Cref{rotate-patch-fig,insert-delete-fig} can be retired after the SCX replaces them.
Safely retiring versions and \PropStatus\ objects, which can be reached via multiple paths from the root, is more difficult.
To do so, we use the property of EBR that an object is safe to retire at time $T$ if it will not be accessed by any high-level operation that starts after time $T$.

We first create separate functions for top-level \op{Refresh}es 
and recursive \op{Refresh}es, 
defined in Section~\ref{sec:delegation}.
Each time a propagate performs a successful top-level \op{Refresh}, it keeps track of the old version that it replaced in a \x{toRetire} list.
No node points to this old version, but
it is not safe to retire yet, since it could still be reachable from the root of the version tree.
Once the \op{Propagate} reaches the root, all versions in its \op{toRetire} list are guaranteed to be unreachable from the root of the version tree and they can therefore be safely retired.

We have shown how to reclaim old versions replaced by newer ones.
What remains is to reclaim the final version stored in each node.
A \op{Refresh} can change a node's version pointer even after the node is retired, but the version pointer stops changing when the node is safe to free.
Even then, the final version may still be accessed by a long-running query via 
an old version tree.
However, newly started queries will not access the final version, so it can be safely retired immediately before freeing the node.

Safely reclaiming \PropStatus\ objects also poses a challenge because each one can be pointed to by multiple versions, and it is not clear when they all become unreachable.
To avoid complex reachability checks, we observe that a \PropStatus\ object can be safely retired at the end of the \op{Propagate} operation $P$ that created it, even while the object is still reachable.
This is because the only operations that access a \PropStatus\ object are those  \crchanges{whose work is delegated to $P$ (directly or indirectly)}. This delegation can only happen while $P$ is running.
Therefore, any high-level operation that starts after $P$ completes will never access $P$'s \PropStatus\ object. 


%% file: sections/experimental_results.tex
\section{Experimental Results} \label{sect:experiments}

\renewcommand{\tabcolsep}{3pt}
\begin{table}
\caption{Summary of data structure properties}
\label{tab:dscompare}
\begin{center}
\small
\begin{tabular}{ |c|c|c|c|c|c|c| } 
\hline
\footnotesize
 & Augmented & Balanced & Fanout & Lock-free \\
\hline
BAT & yes & yes & 2 & yes \\
\hline
BAT-Del & yes & yes & 2 & yes \\
\hline
BAT-EagerDel & yes & yes & 2 & yes \\
\hline
FR-BST~\cite{FR24} & yes & no & 2 & yes \\
\hline
\makecell{Bundled~\cite{NHP22} \\ CitrusTree} & no & no & 2 & no \\
\hline
VcasBST~\cite{WBB+21} & no & no & 2 & yes \\
\hline
VerlibBTree~\cite{blelloch2024verlib} & no & yes & 4-22 & yes \\
\hline
\end{tabular}
\end{center}
\end{table}


\crchanges{We implemented BAT using the chromatic tree implementation in \cite{BER14}, which uses \op{LLX}/\op{SCX} primitives from \cite{BER13}, in the publicly available SetBench~\cite{setbench} microbenchmark.}
\crchanges{Our experiments} show that (1) delegation can improve BAT's insert and delete throughput by up to 200\%, (2) BAT scales well with thread count and data structure size, and (3) BAT performs significantly faster than previous data structures in workloads with more than 0.15-11\% order-statistic queries (depending on data structure size) and in workloads with range queries of size larger than 2K-10K. 
\crchanges{Our code is available at \url{https://github.com/evanwrench/CBAT}.}

{\renewcommand{\arraystretch}{1.1}
\begin{table*}
\caption{Description of our experiments}
\label{tab:description}
\footnotesize
\begin{tabular}{ |m{0.12\textwidth} | m{0.075\textwidth} | >{\raggedright\arraybackslash}m{6cm}| >{\raggedright\arraybackslash}m{6cm}| }

\hline

\textbf{Experiment} & \textbf{Figure} & \textbf{Description} & \textbf{Purpose} \\
\hline
\makecell{Improvement \\Study} & \ref{fig:ablation}, \ref{fig:ben_balance} & Throughput vs number of threads on an update only workload & Compare variants of \abt{} with FR-BST under uniform and skewed workloads  \\
\hline
\makecell{Query\\ Scalability} & \ref{fig:query_graph} & Throughput vs number of threads for rank, select and range queries on BAT-EagerDel & Compare the scalability of different queries on our balanced augmented tree  \\
\hline
\makecell{Range Query \\Size} & \ref{fig:ben_augmentation} & Throughput vs range query size for a 20\% update workload & Show the benefits and tradeoffs of augmenting a concurrent tree \\
\hline
\makecell{Rank Query \\Percentage} & \ref{fig:rank_exp} & Throughput vs percentage of rank queries for varied percentage of updates & Compare the performance of different trees under varying amounts of rank queries \\
\hline
\makecell{Thread \\Scalability} & \ref{fig:scal_thread} & Throughput vs number of threads for different workloads & Show how various trees scale to higher numbers of threads \\
\hline
\makecell{Isolated\\ Performance} & \ref{fig:individual} & Average update/range query time vs range query size for a 20\% update workload & Show the individual performance of updates and queries \\
\hline
\makecell{Size \\Scalability} & \ref{fig:scal_size} & Throughput vs data structure size with keys chosen using a Zipfian distribution & Show how various trees scale to larger data structure sizes \\
\hline
\end{tabular}
\end{table*}
}

\begin{figure*}
    \vspace{-0.4cm}
    \begin{minipage}{0.32\textwidth}
    \includegraphics[width=\textwidth]{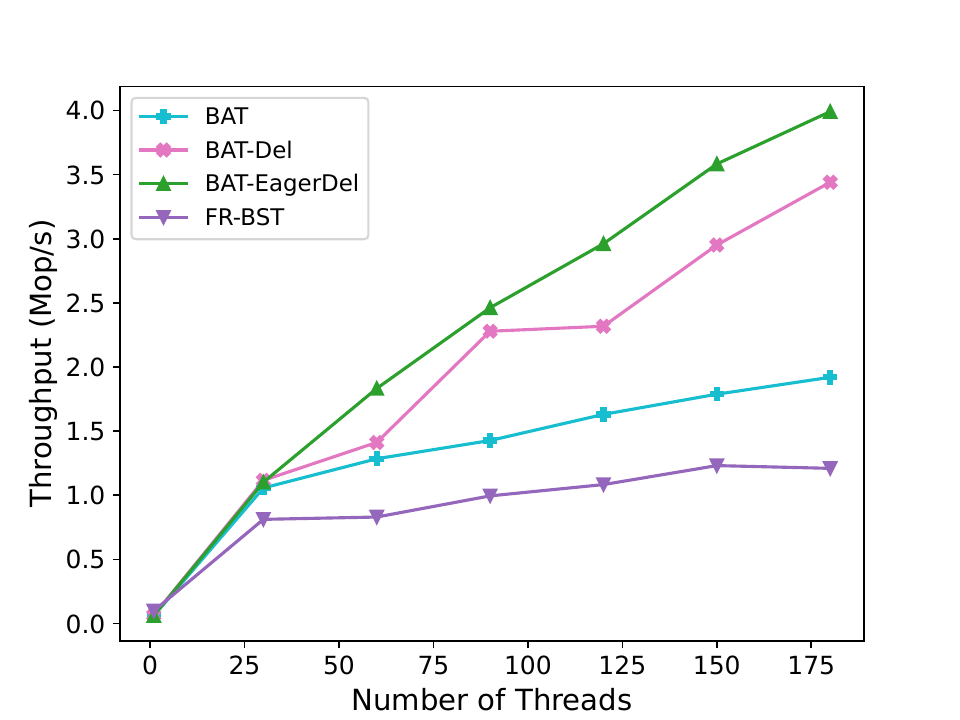}
    \subcaption{MK 10M, 50-50-0-0. Comparing our three variants.}
    \label{fig:ablation}
    \end{minipage}\hfill
    \begin{minipage}{0.32\textwidth}
    \includegraphics[width=\textwidth]{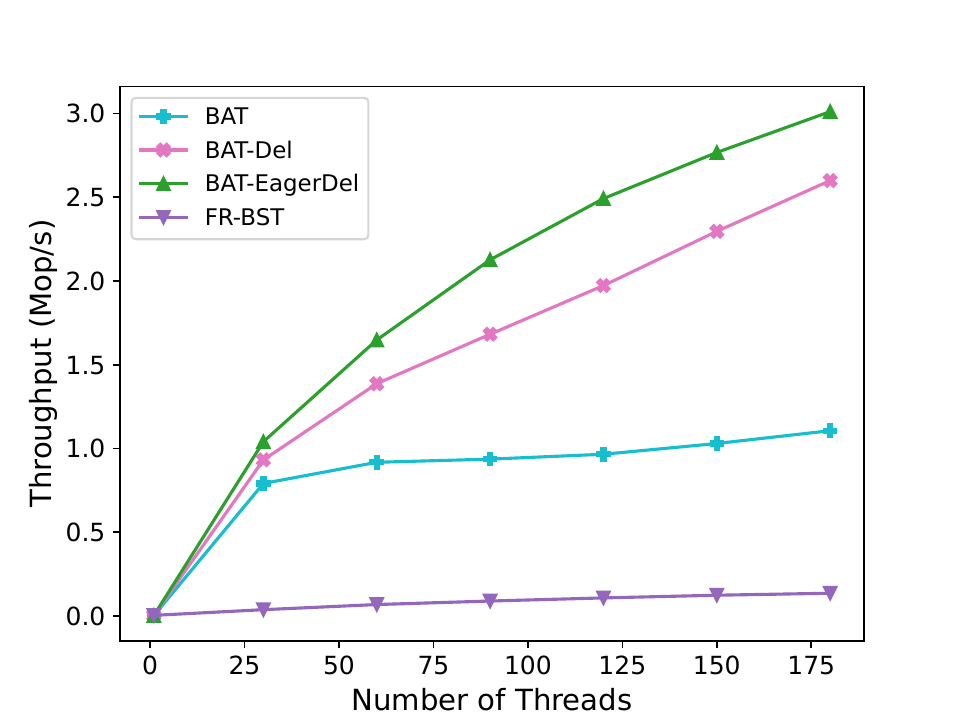}
    \subcaption{MK 10M, 100-0-0-0, Sorted distribution. Benefits of balancing BST.}
    \label{fig:ben_balance}
    \end{minipage}\hfill
    \begin{minipage}{0.32\textwidth}    
    \includegraphics[width=\textwidth]{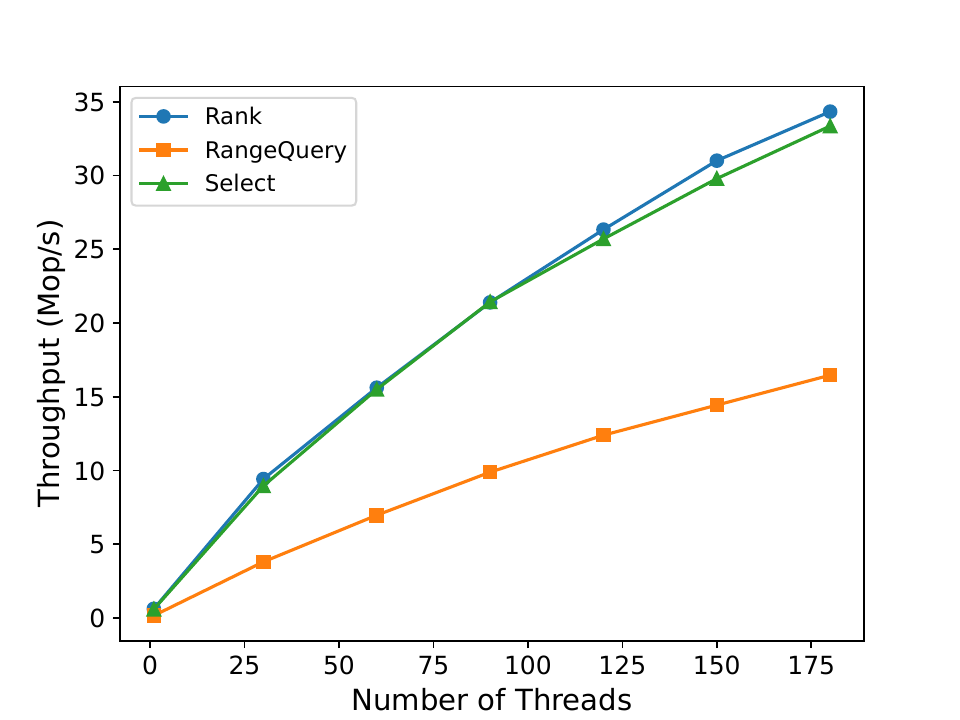}
    \subcaption{RQ 50K, MK 10M 5-5-0-90. Scalability of queries on BAT-EagerDel.}
    \label{fig:query_graph}
    \end{minipage}
    \caption{Performance of variants of \abt.}
    \label{fig:internal_graphs}
\end{figure*}

\myparagraph{Data Structures.} The table~\ref{tab:dscompare} summarizes the key properties of the data structures we compare.
All are linearizable, written in C++, integrated into SetBench, and use the same memory reclamation scheme.
Data structures that do not maintain augmented values achieve linearizable range queries by taking a snapshot and iterating over the range.
Each implementation (other than FR-BST, which we implemented) was taken from its original paper.



\myparagraph{Setup.}
Our experiments ran on a 96-core Dell PowerEdge R940 machine with
4x Intel(R) Xeon(R) Platinum 8260 CPUs (24 cores, 2.4GHz and 143MB L3 cache each),
and 3.7TB memory. 
Each core is 2-way hyperthreaded, giving 192 hyperthreads.
We used \texttt{numactl -i all}, 
evenly spreading the memory pages across the sockets in a round-robin fashion.
The machine runs Ubuntu 22.04.5 LTS.
The C++ code was compiled with g++ 11.4.0 with \texttt{-O3}. 
For scalable memory allocation, we used mimalloc~\cite{mimalloc}. 
Memory was reclaimed using DEBRA~\cite{brown2015reclaiming}, an optimized implementation of epoch-based reclamation~\cite{fraser2004practical} (see Section~\ref{sec:memory-reclamation} for more details).
\crchanges{BAT used the original \op{LLX}/\op{SCX} implementation as described in \cite{BER13}.  The optimized \op{LLX}/\op{SCX} implementation of~\cite{arbel2017reuse} could provide additional improvements}.

All of our experiments (other than Figure~\ref{fig:ben_balance}, which has no prefilling) began with a prefilling phase where random inserts and deletes ran until the structure \crchanges{contained half the keys in the key range.  
Then, threads perform operations chosen randomly (using various 
distributions, described below) for 3 seconds.
We report the average of 5 runs.}
The variance within trials of the same experiment was relatively consistent between experiments of different parameters, with the lowest throughput of a trial being within around 8--10\% of the highest throughput of a trial. Since most of the plots we show are on a log scale, this difference is hardly visible.

\myparagraph{Workloads.} We varied the following parameters:

\noindent \textit{Total Threads} (TT): Number of threads concurrently executing operations on the data structure.

\noindent \textit{Max Key} (MK): The maximum integer key that can be inserted into the data structure. Since all our experiments \crchanges{(except Figure \ref{fig:ben_balance})} were run with the same fraction of inserts as deletes and the trees are pre-filled with half the key range, the size of the data structure remained around half the size of this parameter.

\noindent \textit{Range Query Size} (RQ): The size of each range query performed. The lower bound of the range query interval was generated uniformly from the range of valid lower bounds and is added to this parameter to get the upper bound.

\noindent \textit{Workload} (i\%-d\%-f\%-rq\%): The probability of choosing each operation (insert, delete, find, range query) as the next one a thread executes. In Figure~\ref{fig:query_graph}, rq\% is replaced with the percentage of the given query (rank, select, or rangeQuery). In Figure~\ref{fig:rank_exp}, rq\% is replaced with rank\%.

\noindent \textit{Key Distribution}: Either uniform, sorted or Zipfian.
The sorted and Zipfian workloads result in high contention as updates are routed to the same parts of the tree.
The  distribution was uniform unless otherwise specified.

\crchanges{The sorted distribution inserted keys in roughly
increasing order to evaluate the benefits of balancing (Figure~\ref{fig:ben_balance}).
Threads acquired} keys to insert from an increasing global counter.
To reduce contention on the counter, threads incremented it by 100 each time to acquire a batch of 100 keys.
    


\begin{figure*}
\vspace{-0.4cm}
\begin{minipage}{0.32\textwidth}
    \begin{minipage}[t]{0.96\textwidth}    
    \includegraphics[width=\textwidth]{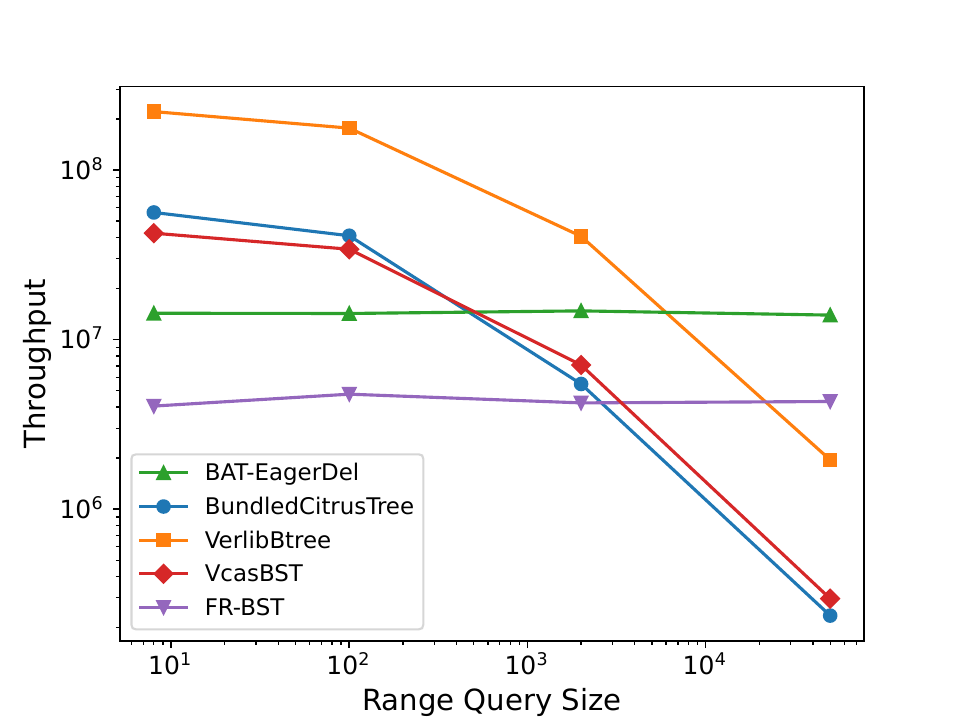}
    \subcaption{TT 120, MK 100K, 10-10-40-40. Benefits of augmenting BST, small tree.\label{fig:ben_aug_small}}
    \end{minipage}
    \\
    \begin{minipage}[t]{0.96\textwidth}    
    \includegraphics[width=\textwidth]{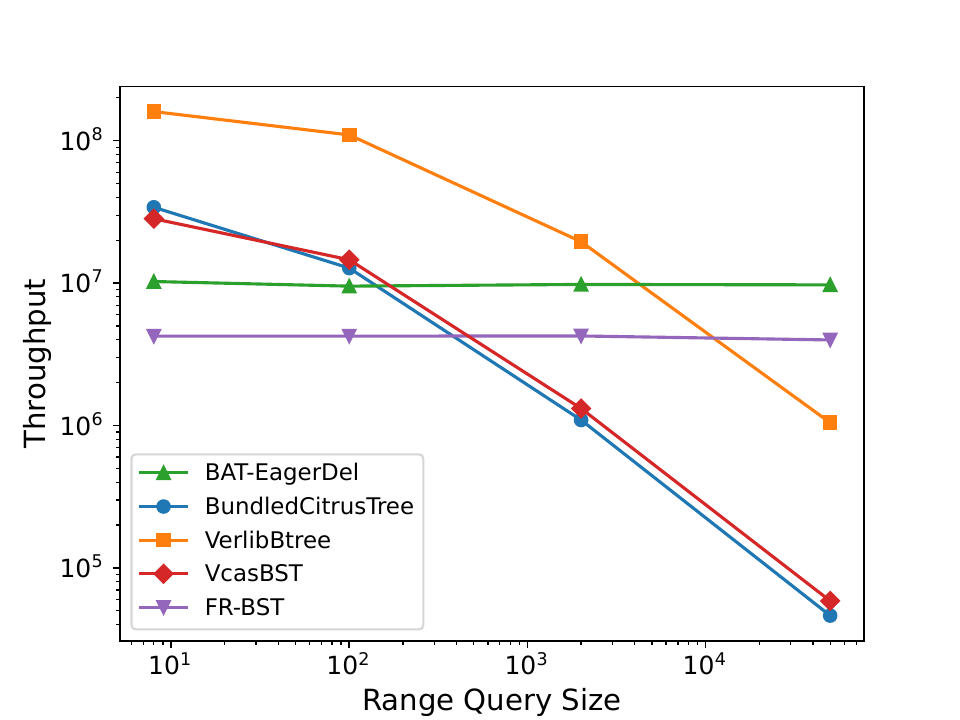}
    \subcaption{TT 120, MK 10M, 10-10-40-40. Benefits of augmenting BST, large tree.\label{fig:ben_aug_large}}
    \end{minipage}
\caption{Performance of our top performing BAT with respect to range query size.\label{fig:ben_augmentation}}
\end{minipage}
\hfill
\begin{minipage}{0.32\textwidth}
    \begin{minipage}[t]{0.96\textwidth}    
    \includegraphics[width=\textwidth]{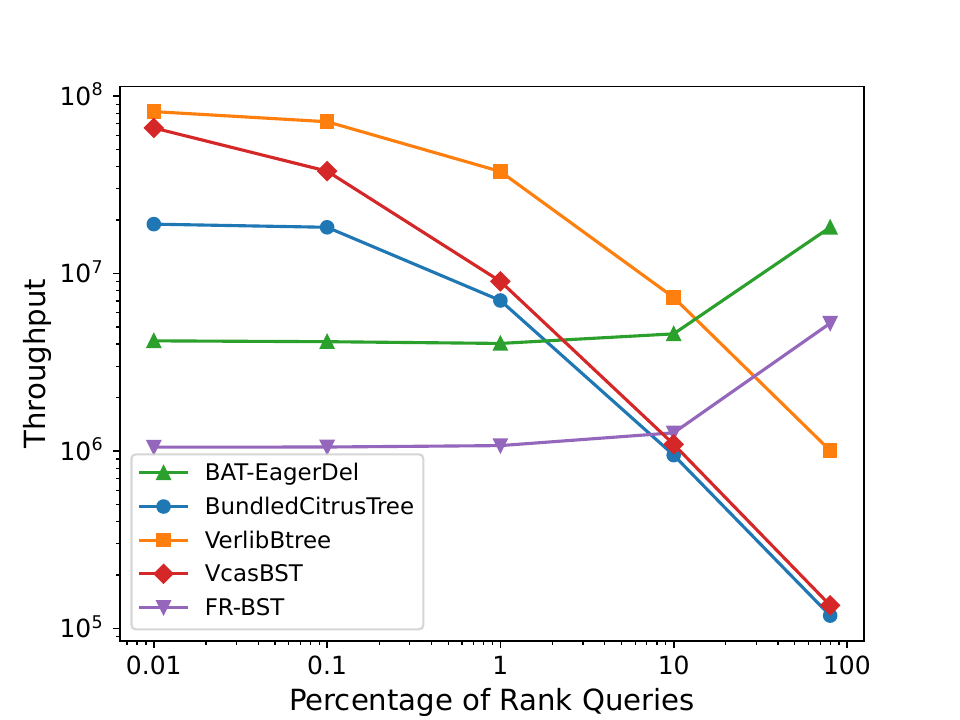}
    \subcaption{TT 120, MK 100K, $\frac12$(100-x)-$\frac12$(100-x)-0-x. Performance on small tree for x\% of rank queries.\label{fig:rank_exp_small}}
    \end{minipage}\hfill
    \\
    \begin{minipage}[t]{0.96\textwidth}    
    \includegraphics[width=\textwidth]{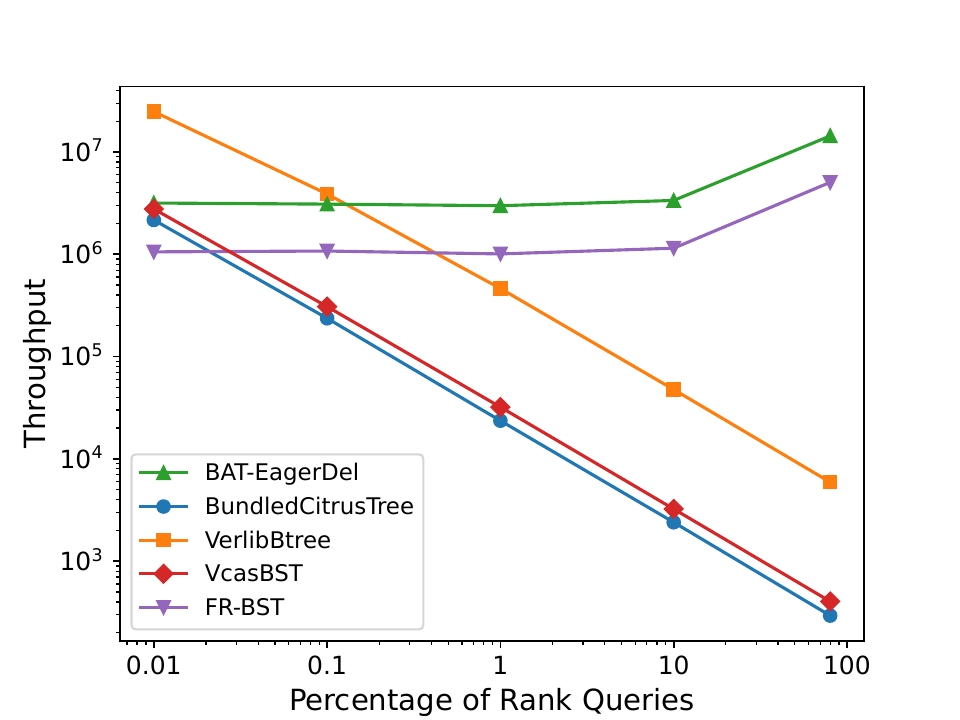}
    \subcaption{TT 120, MK 10M, $\frac12$(100-x)-$\frac12$(100-x)-0-x. Performance on large tree for x\% of rank queries.\label{fig:rank_exp_large}}
    \end{minipage}
    \caption{Performance of our top performing \abt\ on different workloads of rank queries.\label{fig:rank_exp}}
\end{minipage}
\hfill
\begin{minipage}{0.32\textwidth}
    \begin{minipage}[t]{0.96\textwidth}    
    \includegraphics[width=\textwidth]{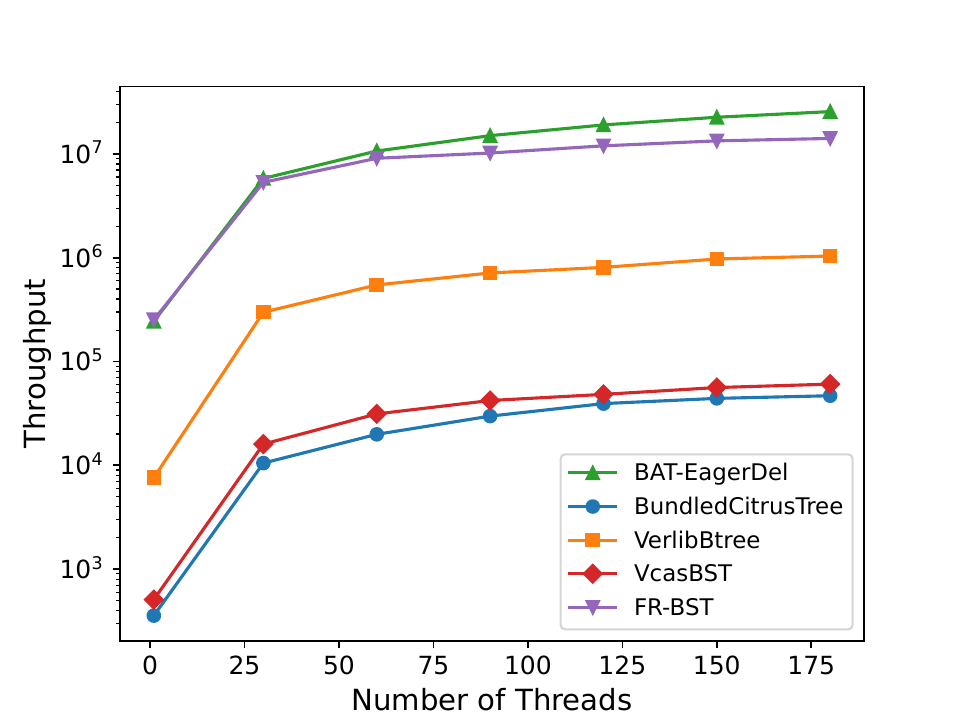}
    \subcaption{RQ 50K, MK 10M, 2.5-2.5-47.5-47.5. Thread scalability, low update workload.\label{fig:scal_thread_low}}
    \end{minipage}\hfill
    \\
    \begin{minipage}[t]{0.96\textwidth}    
    \includegraphics[width=\textwidth]{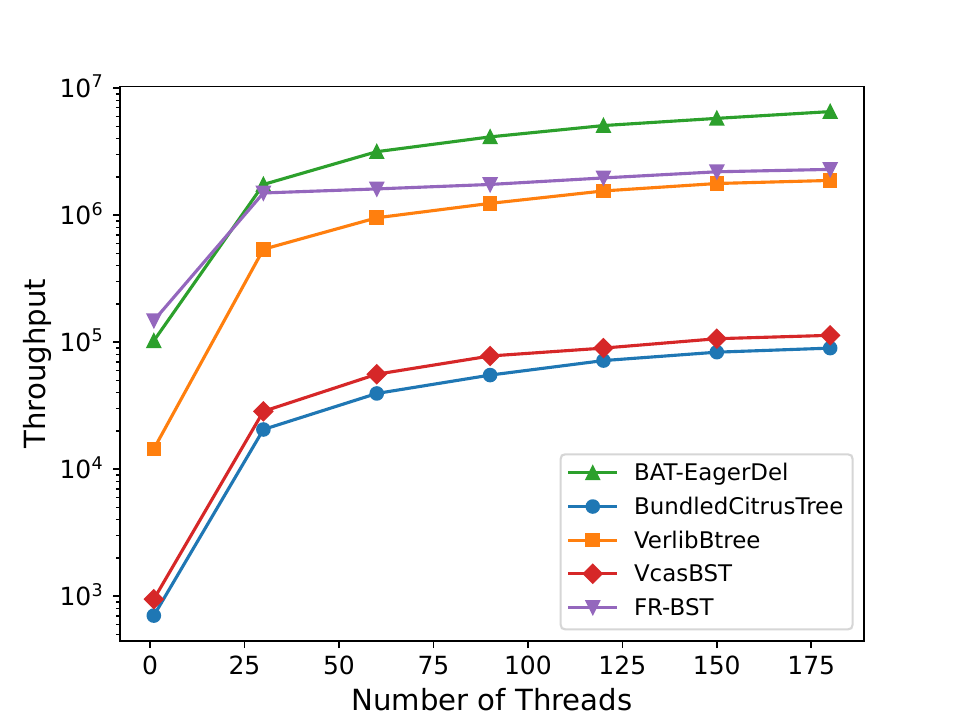}
    \subcaption{RQ 50K, MK 10M, 25-25-25-25. Thread scalability, high update workload.\label{fig:scal_thread_high}}
    \end{minipage}
        \caption{Performance of our top performing \abt\ with respect to number of threads.\label{fig:scal_thread}}
\end{minipage}
\end{figure*}

\begin{figure*}
    \begin{minipage}{0.64\textwidth}
    \begin{minipage}{0.48\textwidth}    
    \includegraphics[width=\textwidth]{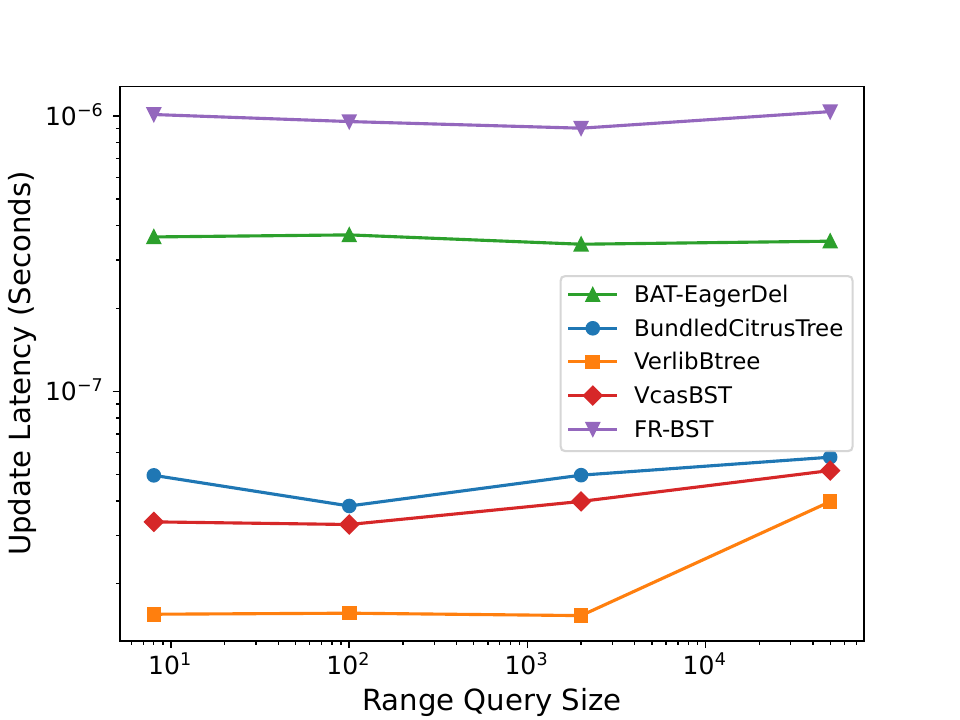}
    \subcaption{TT 120, MK 10M, 10-10-40-40. Average update latency.}
    \label{fig:update_latency}
    \end{minipage}\hfill
    \begin{minipage}{0.48\textwidth}
    \includegraphics[width=\textwidth]{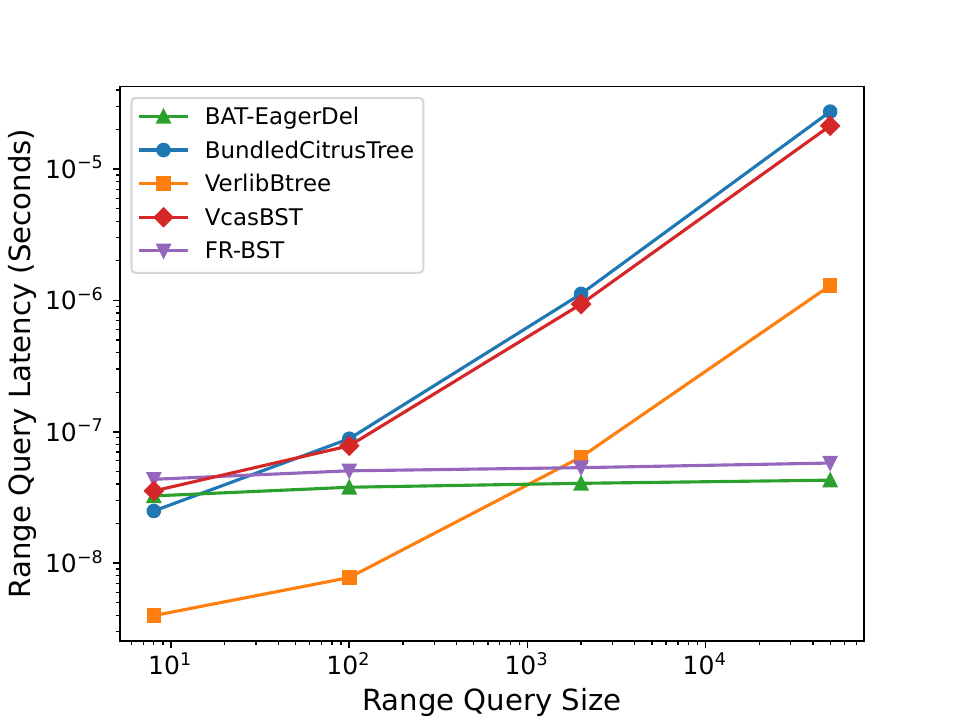}
    \subcaption{TT 120, MK 10M, 10-10-40-40. Average range query latency.}
    \label{fig:rq_latency}
    \end{minipage}\hfill
    \caption{Performance of updates and range queries with respect to range query size on a mixed workload.}
    \label{fig:individual}
    \end{minipage}\hfill
    \begin{minipage}{0.32\textwidth}
    \begin{minipage}{0.96\textwidth}
    \includegraphics[width=\textwidth]{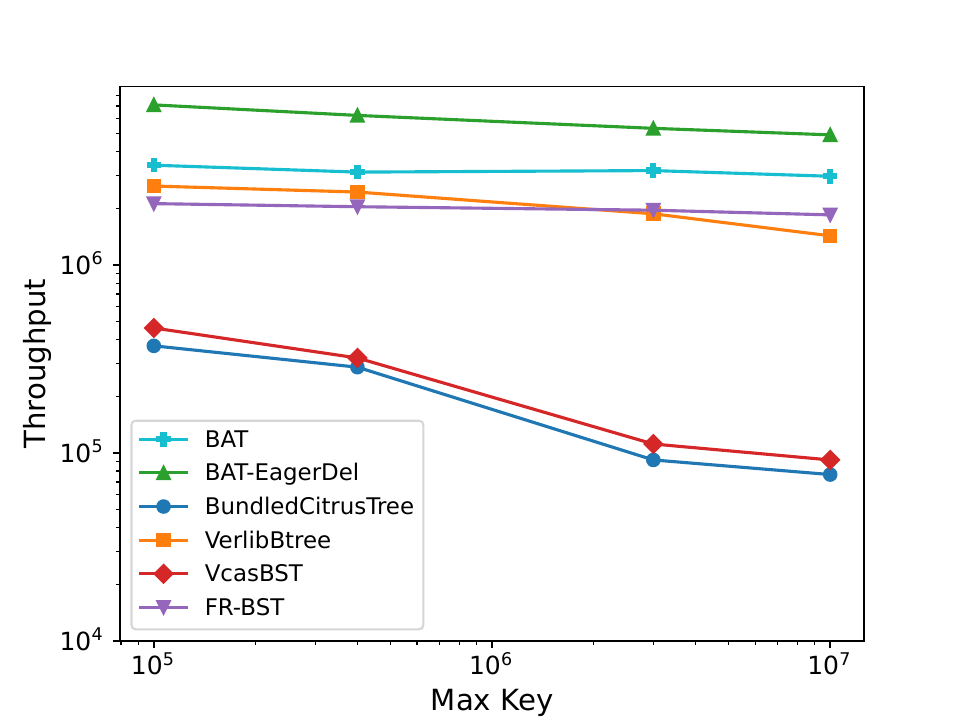}
    \end{minipage}
    \caption{Comparison of \abt\ variants to other trees with respect to data structure size. TT 120, RQ 50K, 25-25-25-25. Size scalability, high update Zipfian \crchanges{(parameter=0.95)} workload.}
    \label{fig:scal_size}
    \end{minipage}
\end{figure*}

\myparagraph{Results. } \crchanges{We summarize our experiments in} Table \ref{tab:description}. \crchanges{Each entry describes one type of experiment and its purpose.} 

\myparagraph{\crchanges{Comparing Augmented Trees.}} Figures \ref{fig:ablation} and \ref{fig:ben_balance} show the performance improvements for updates that we get from variants of our algorithm under two different workloads. 
As expected, balancing allows BAT to significantly outperform \crchanges{the unbalanced} FR-BST, especially when using a sorted workload. 
The average number of nodes seen by a \op{Propagate} decreases from 31 (in FR-BST) to 25 (in BAT) in the uniform workload (Figure~\ref{fig:ablation}, 180 threads) and 2300 to 56 in the sorted workload (Figure~\ref{fig:ben_balance}, 180 threads). 

Adding delegation also improves throughput by around 100\% in the case of \delegateTwoF{} and 120\% in the case of \delegateOneF{} for update-only uniform workloads on 180 threads. This is because delegation reduces the average number of nodes \crchanges{a propagate visits in} a tree with 5M keys by around 3 for 
\crchanges{delegation after two failed \op{Refresh} in \delegateTwoF{}} 
and 4.5 for 
\crchanges{delegation after single failed \op{Refresh} in \delegateOneF{}}.
Since these nodes are usually close to the top of the tree, this greatly reduces the bottleneck at these levels.
We therefore focus on the \delegateOneF{} variant for the remaining comparisons.


\myparagraph{Queries.} In Figure~\ref{fig:query_graph}, we see the performance of several order-statistic queries on BAT-EagerDel \crchanges{scales well}. Rank queries return the number of keys in the set that are less than or equal to a given key.
Select queries return the $k$th smallest key in the set, for a given $k$. Range queries return the number of keys in a given range. They are slower than rank and select queries due to having to traverse two paths (for the lower and upper bound of the range) in the BST instead of just one. 



\myparagraph{Range Query Size.} Figure~\ref{fig:ben_augmentation} shows the performance of various structures under varying range query size. 
\crchanges{Figures \ref{fig:ben_aug_small} and \ref{fig:ben_aug_large} show results for small and large trees.}
Since the unaugmented trees perform work proportional to the number of keys in the range, their performance drops off sharply for larger range queries. 
In contrast, in the augmented trees (FR-BST and ours), queries
only perform work proportional to the height of the BST, so their performance stays consistent no matter the range query size.
VerlibBtree outperforms the other non-augmented trees since it uses higher fanout trees for better cache efficiency, but loses out to BAT-EagerDel after range query sizes reach 2000--4000. After reaching range query size 2M for a tree of size 10M, BAT-EagerDel is ~400x as fast as the closest non-augmented tree. However, the added overhead for inserts and deletes causes the augmented structures to lose out heavily when range queries only traverse a few keys. For range queries of only 8 keys, BAT-EagerDel is 15x slower than VerlibBtree. 
BAT-EagerDel is around 3x faster than FR-BST because balancing reduces the average depth of the leaves and delegation reduces contention at higher levels of the tree.

\myparagraph{Rank Queries.} Figure~\ref{fig:rank_exp} compares the performance of concurrent trees with different percentages of rank queries. We only use rank and not select here since the similarity in their algorithms would produce an identical graph for select. \crchanges{We vary the percentage of rank queries and the remaining operations are split evenly between inserts and deletes (e.g., 1\% rank, 49.5\% insert, 49.5\% delete).}
Since non-augmented rank queries take time proportional to the number of keys less than the selected key, the downside of non-augmented trees is less pronounced in smaller trees \crchanges{(Figure \ref{fig:rank_exp_small})}. However, BAT-EagerDel still performs best for more than ~11\% rank queries. In larger trees \crchanges{(Figure \ref{fig:rank_exp_large})}, we can see BAT-EagerDel outperforms the other structures even for 0.15\% rank queries. \crchanges{We see a large improvement in BAT-EagerDel and FR-BST when going from 10\% to 80\% rank queries since there is a significant drop in the number of inserts and deletes, which are the worst performing operations for these structures.}

\myparagraph {Thread Scalability.}  
Figures \ref{fig:scal_thread_low} and \ref{fig:scal_thread_high} show scalability under  low (5\%) and high (50\%) update percentages, respectively. 
These update percentages were selected according to YCSB workloads A and B~\cite{BFC10}. \crchanges{FR-BST scales less well than other data structures since it has higher contention close to the root, making updates perform worse when more threads are involved.} The scaling of the other structures are similar to each other, however BAT-EagerDel outperforms the closest \crchanges{unaugmented} competitor by around 4x on the high update workload and 30x on the low update workload
for all threads.

\myparagraph{Isolated Performance.} Figure~\ref{fig:individual} shows the average latency in seconds for updates and range queries under the same workload as Figure ~\ref{fig:ben_aug_large}. In the update graph, we see that the performance of inserts and deletes on BAT-EagerDel remains relatively constant. Furthermore, BAT-EagerDel has a lower range query latency than all unaugmented competitors when at least 2000 elements are included in the query.

\myparagraph{Size Scalability and Zipfian Distribution.} 
Figure~\ref{fig:scal_size} shows the effect of increasing the data structure size (by varying the maximum key). \crchanges{We include the BAT variant with no delegation in this graph to show that delegation still has benefits when keys are chosen from a Zipfian distribution.} Overall, we see that BAT-EagerDel scales slightly better with size compared to VerlibBtree, BundledCitrusTree and VcasBST.
\crchanges{We only show results for Zipfian distribution, as the results for the uniform distribution are nearly identical.}

\myparagraph{Why Balancing Improves Throughput.} \crchanges{
BAT performs extra work to balance the tree, and calls to propagate
occasionally have to retraverse parts of the tree or fill in \nil\ versions.
Nevertheless, our results show that BAT variants consistently outperform FR-BST, even in workloads where FR-BST is balanced in expectation.
We provide some key statistics to explain this. We measured on a workload with 120 threads, 100K max key, 50K range query size and an even percentage of inserts, deletes, range queries and finds on both a uniform and Zipfian distribution with parameter 0.99. On the BAT variants, each propagate only traverses around 6.4\% (5.9\%) more nodes beyond the initial search path for the uniform (Zipfian) distribution. A call to propagate fills in only 0.075 (0.03) \nil\ versions on average. Lastly, the average number of \op{CAS}es attempted during a propagate call is 22.2 (22.4) for BAT, 13.9 (13.2) for BAT-EagerDel and 26.8 (27.5) for FR-BST. 
Thus, the extra costs incurred by rebalancing are minimal compared to the advantages of maintaining a more carefully balanced tree.
}

%% file: sections/conclusion.tex
\section{Conclusion}

Augmentation makes search trees significantly more versatile by extending the set interface to enable support for aggregation queries, order-statistic queries, and range queries.
In this paper, we designed, implemented, and empirically validated \abt, the first lock-free Balanced Augmented Tree supporting generic
augmentation functions. 
While we emphasized our augmentation scheme as applied to a chromatic tree, our scheme is general---adaptable to concurrent search trees where updates modify the tree by replacing one patch by another patch of new nodes.

Our experiments show that \abt\ and its optimized versions  
are scalable.
For applications where augmentation is essential,
\abt\ is the only efficient, concurrent option to date.
Some queries, like finding the predecessor of a given key, can be answered by exploring a small part of a snapshot of the tree.  
In such cases, snapshots, e.g. \cite{WBB+21, JJ25}, provide a sufficiently good solution because they avoid the overhead of augmentation.
However, our experiments show that queries that have to traverse many nodes of the tree---like range queries, rank queries or selection queries---are vastly faster with \abt\ than with other snapshot-based approaches.  
Thus, even if the workload is mostly updates with occasional queries, \abt\ outperforms other approaches.

There remain interesting open directions in designing concurrent search tree data structures.
Complex sequential data structures like link/cut trees \cite{ST83}, measure trees \cite{GMW83}, and priority search trees \cite{McC85} rely on balanced augmented trees.
Now that we have designed a {\em concurrent} balanced augmented tree, we can ponder the possibility of concurrent versions of these more complex data structures.



\section*{Acknowledgements}

This research was funded by the Natural Sciences and Engineering Research
Council of Canada, Dartmouth College, and the Greek Ministry of Education, Religious Affairs and Sports call SUB 1.1--Research Excellence Partnerships (Project:\linebreak HARSH, code: Y$\Pi$ 3TA-0560901), implemented through the National Recovery and Resilience Plan Greece 2.0 and funded by the European Union--NextGenerationEU. 
We thank the anonymous reviewers for their feedback, which helped improve the manuscript.

%% file: sections/algorithm_delegate.tex
\newpage
\section{Algorithm with Delegate Mechanism}
\label{sec:delegate-alg}

Here, we give the details of the delegation mechanisms described in \Cref{sec:delegation}.
Figure~\ref{delegate1-types} and~\ref{1F-helper} shows code that is common to both \delegateTwoF{} and \delegateOneF{}.
The Propagate functions of \delegateTwoF{} and \delegateOneF{} are shown in Figures~\ref{2F-pseudocode} and~\ref{1F-pseudocode} respectively.



\begin{figure}
\begin{algorithmic}[1]
\State \textbf{type} Version
\State \hspace{1em} Version* \f{left}, \f{right}
\State \hspace{1em} Key $k$
\State \hspace{1em} int \f{size}
\State \hspace{1em} \hl{PropStatus* \textit{status}}

\medskip

\State \textbf{type} \PropStatus
\State \hspace{1em} Boolean \f{done} 
\State \hspace{1em} \PropStatus* \f{delegatee}
\end{algorithmic}
\caption{Modification to Version object for \delegateTwoF and \delegateOneF, and the new \PropStatus\ object. Nodes are as described in \Cref{abt-pseudocode}.\label{delegate1-types}}
\end{figure}

\begin{figure}
\begin{algorithmic}[1]
\Function{WaitForDelegatee}{\PropStatus* $d$}{}
    \While{$\neg d.\f{done}$}
        \If{$d.\f{delegatee} \neq \nil$}
            \State $d \gets d.\f{delegatee}$
        \EndIf
    \EndWhile
\EndFunction{WaitForDelegatee}
\medskip
\Function{ReadVersion}{Node* $x$}{Version*}
    \LComment{Sets $x.version$ if \nil\ and then returns $x.version$.}
    \State Version* $\x{v} \gets \x{x}.\f{version}$
    \If{$v = \nil$} 
        \State \op{RefreshNil}($x$)  
        \State $v \leftarrow x.version$
    \EndIf
    \State \Return $v$
\EndFunction{ReadVersion}
\medskip
\Function{RefreshNil}{Node* $x$}{}
    \LComment{Recursive refresh for setting \nil\ versions.}
    \Repeat
        \State Node* $x_l \gets \x{x}.\f{left}$
        \State Version* $v_l \leftarrow \op{ReadVersion}(x_l)$
    \Until{$x_l = \x{x}.\f{left}$}
    \Repeat
        \State Node* $x_r \gets \x{x}.\f{right}$
        \State Version* $v_r \leftarrow \op{ReadVersion}(x_r)$ 
    \Until{$x_r = \x{x}.\f{right}$}

    \State Version* $new\! \gets\!$ new Version($k \gets \x{x.k}, \f{left} \gets v_l,$  
    \Statex \hfill $\f{right} \gets v_r, \f{size} \gets v_l.\f{size} + v_r.\f{size},\x{status} \gets \bot$)
    \State $\op{CAS}(x.\f{version}, \nil, \x{new})$ 
\EndFunction{RefreshNil}
\medskip
\Function{Refresh}{Node* $x$, \PropStatus* \x{ps}}{Boolean, \PropStatus*, Version*, Version*}
    \LComment{Return True if Refresh succeeds, False otherwise}
    \LComment{Also returns \PropStatus\ of propagate that blocked}
    \LComment{the CAS (or \nil\ if successful).}
    \LComment{Also returns left and right versions that were read}
    
    \State Version* $\x{old} \leftarrow \op{ReadVersion}(x)$
    \Repeat
        \State Node* $x_l \gets \x{x}.\f{left}$
        \State Version* $v_l \leftarrow \op{ReadVersion}(x_l)$
    \Until{$x_l = \x{x}.\f{left}$}
    \Repeat
        \State Node* $x_r \gets \x{x}.\f{right}$
        \State Version* $v_r \leftarrow \op{ReadVersion}(x_r)$ 
    \Until{$x_r = \x{x}.\f{right}$}

    \State Version* $new\! \gets\!$ new Version($k \gets \x{x.k}, \f{left} \gets v_l,$  
    \Statex \hfill $\f{right} \gets v_r, \f{size} \gets v_l.\f{size} + v_r.\f{size},\x{status} \gets \x{ps}$)
    \State Version* $res \leftarrow \op{CAS}(x.\f{version}, \x{old}, \x{new})$ 
    \State Boolean $\x{success} \gets (\x{res} = \x{old})$
    \State \Return $\x{success}, (\x{success}\ ?\ \nil : \x{res}.\f{status})$, $v_l$, $v_r$
\EndFunction{Refresh}
\end{algorithmic}
\caption{Helper functions for \delegateTwoF\ and \delegateOneF\label{1F-helper}}
\end{figure}

\begin{figure}
\begin{algorithmic}[1]
\Function{Propagate}{Key $k$}{}
    \State Set $\x{refreshed}\leftarrow \{\}$\Comment{stores refreshed nodes}
    \State Stack $\x{stack}$ initialized to contain \x{Root} \Comment{thread-local}
    \State \PropStatus* $\x{ps} \gets $ new \PropStatus($\f{done}\gets \false,$ 
    \Statex \hfill                                     $\f{delegatee}\gets \nil$)
    \Repeat 
        \State Node* $\x{next}\leftarrow \x{stack}.\op{Top}()$
        \Loop \Comment{go down tree until child is refreshed}
            \State $\x{next}\!\leftarrow\! (k < \x{next}.\f{key}\, ?\, \x{next}.\f{left}\! :\! \x{next}.\f{right})$
            \State \textbf{exit when} $\x{next}\in\x{refreshed}$ or \x{next} is a leaf
            \State \x{stack}.\op{Push}(\x{next})
        \EndLoop
        \State Node* $\x{top}\leftarrow \x{stack}.\op{Pop}()$
        \State $\x{success}, *, *, * \gets \op{Refresh}(\x{top}, \x{ps})$ \label{lin:2f-refresh1}
        \If{$\neg \x{success}$}\Comment{if try1 fails} 
            \State $\x{success},\x{del},*,* \gets \op{Refresh}(\x{top}, \x{ps})$
            \If{$\neg \x{success}$ and $\neg \x{top}.\x{finalized}$} 
                \State $\x{ps}.\f{delegatee} \gets \x{del}$
                \State \op{WaitForDelegatee}(\x{ps}.\f{delegatee})
                \LComment{Can be made lock-free by resuming}
                \LComment{from line~\ref{lin:2f-refresh1} after waiting exceeds}
                \LComment{a time limit.}
                \State $\x{ps}.\f{done} \gets \true$
                \State \Return
            \EndIf
        \EndIf
        \State $\x{refreshed} \gets \x{refreshed} \cup \{\x{top}\}$
    \Until{$\x{Root}\in\x{refreshed}$}
    \State $\x{ps}.\f{done} \gets \true$
\EndFunction{Propagate}
\end{algorithmic}
\caption{\delegateTwoF\label{2F-pseudocode}}
\end{figure}

\begin{figure}
\begin{algorithmic}[1]
\Function{Propagate}{Key $k$}{}
    \State Set $\x{refreshed}\leftarrow \{\}$\Comment{stores refreshed nodes} 
    \State Stack $\x{stack}$ initialized to contain \x{Root} \Comment{thread-local}
    \State \PropStatus* $\x{ps} \gets $ new \PropStatus($\f{done}\gets \false,$ 
    \Statex \hfill                                     $\f{delegatee}\gets \nil$)
    \Repeat 
        \State Node* $\x{next}\leftarrow \x{stack}.\op{Top}()$
        \Loop \Comment{go down tree until child is refreshed}
            \State $\x{next}\!\leftarrow\! (k < \x{next}.\f{key}\, ?\, \x{next}.\f{left}\! :\! \x{next}.\f{right})$
            \State \textbf{exit when} $\x{next}\in\x{refreshed}$ or \x{next} is a leaf
            \State \x{stack}.\op{Push}(\x{next})
        \EndLoop
        \State Node* $\x{top}\leftarrow \x{stack}.\op{Pop}()$
        \Repeat \label{lin:1f-repeat}
            \State $\x{success},\x{del}, v_l, v_r$ $\gets \op{Refresh}(\x{top}, \x{ps})$
            \If{$\neg \x{success}$ and $\neg \x{top}.\x{finalized}$} 
                \State $\x{ps}.\f{delegatee} \gets \x{del}$
                \State \op{WaitForDelegatee}(\x{ps}.\f{delegatee})
                \LComment{Can be made lock-free by resuming}
                \LComment{from line~\ref{lin:1f-repeat} after waiting exceeds}
                \LComment{a time limit.}
                \State $\x{ps}.\f{done} \gets \true$
                \State \Return
            \EndIf
        \Until{$\x{success}$ and $v_l = \x{top}.\x{left}.\x{version}$ and
        \Statex  \hspace{1.85cm} $v_r = \x{top}.\x{right}.\x{version}$} \label{lin:1f-repeat-end}
        \State $\x{refreshed} \gets \x{refreshed} \cup \{\x{top}\}$
    \Until{$\x{Root}\in\x{refreshed}$}
    \State $\x{ps}.\f{done} \gets \true$
\EndFunction{Propagate}
\end{algorithmic}
\caption{\delegateOneF\label{1F-pseudocode}, only lines ~\ref{lin:1f-repeat}-\ref{lin:1f-repeat-end} changed relative to Figure~\ref{2F-pseudocode}.}
\end{figure}


%% file: sections/correctness.tex
\clearpage
\section{Correctness}
\label{sec:correctness}
We follow the arguments similar to those for augmented binary search trees in~\cite{FR24} and extend them for the augmented chromatic trees presented in this paper. We first present the proof for the non delegating version. 

\subsection{Facts About the Unaugmented Chromatic Tree}

We first summarize some facts from ~\cite{Bthesis17} about the original, unaugmented, lock-free chromatic tree. Since our augmentation does not affect the \textit{node tree}, these facts remain true in the augmented chromatic tree. 

In the chromatic tree, the coordination of updates using LLX/SCX primitives ensure the following claims.

The following is a consequence of Lemma 3.94 and Lemma 5.1, claim 3, and Corollary 5.2 of~\cite{Bthesis17}.
A Node is considered reachable if it can be accessed by traversing pointers starting from \textit{root}.

\begin{lemma}
    A Node's child pointer can change only when the Node is not finalized and it is reachable.
    \label{lem:lem13}
\end{lemma}

The following is a consequence of Lemma 6.3.3. of~\cite{Bthesis17}.
\begin{lemma}
    If a Node is on the search path for key k in one configuration and is still reachable in some later configuration, then it is still on the search path for k in the later configuration.
    \label{lem:lem14}
\end{lemma}


The chromatic tree uses an ordinary BST search and the following lemma is a direct consequence of Lemma 6.3.4 of~\cite{Bthesis17}.
\begin{lemma}
    If an insert, delete or rebalance operation visits a Node x during its search for the location of key k, then there was a configuration between the beginning of the operation and the time it reaches x when x was on the search path for k in the \textit{node tree}.
    \label{lem:reachability}
\end{lemma}

As shown in Figure 6.3 of~\cite{Bthesis17}, \textit{entry} is a special pointer serving as the immutable root of the \textit{node tree}.

Let $T_C$ be the \textit{node tree} in configuration $C$. 
Let $n$ be the number of keys in the \textit{node tree} and $c$ be the number of pending update operations.

The following directly follows from Lemma 6.3.7 of~\cite{Bthesis17}. 
The Lemma 6.3.7 of~\cite{Bthesis17} implies that the \textit{node tree} is a BST with additional properties required for a chromatic tree.

\begin{lemma}
    For all configurations $C$, $T_C$ is a balanced BST of height $O(\log n+c)$.
    \label{lem:staysCT}
\end{lemma}

For a \textit{node tree}, how the augmentation information propagates up the tree is crucial for correctness. To achieve this, we introduce the notion that describes when an update operation's augmentation information is reflected at a Node in the \textit{node tree}, referred to as the \textit{arrival point} of the update.

\subsection{Linearization Respects Real-Time Order.}
\label{sec:proofp1}
In this section, we begin by formally defining arrival point of an update at a Node.
Then, for an update operation on a given key, we show that the update's \op{Propagate} ensures that the update has an arrival points at each reachable Node on which it performs a double \op{Refresh}. Moreover, that arrival point is during the update's execution interval. 
Eventually, if the call to \op{Propagate} completes, the update is assigned an arrival point at the root, before it returns.
The arrival point of an update at the root is the linearization point of the update.
Each query is also assigned a linearization point when it reads the version pointer of the root Node to get an immutable snapshot of the \textit{version tree} rooted at root.version.
Since arrival at the root serves as the linearization point of the update, and each query is also assigned a linearization point during the query, it follows that the linearization respects the real-time order of operations.

Intuitively, the arrival point of an update operation \x{op} on key \x{k} at a Node \x{x} is the moment in time during its execution when both (a) \x{x} is on the search path for \x{k} and (b) the effect of \x{op} is reflected in the \textit{version tree} rooted at \x{x.version}.

We now formally define arrival points of insert and deletes operations over an execution $\alpha$ of the implementation.

\begin{definition}
The base case defines the arrival points of unsuccessful \op{Insert} and \op{Delete} operations at a leaf.

\begin{enumerate}
\item\label{ap-failed-delete}
A \op{Delete}($k$) whose traversal of the tree ends at a leaf $\ell$ that does not contain $k$ returns false.
Similarly, an \op{Insert}($k$) that reaches a leaf $\ell$ containing $k$ returns false. 
In both cases, \op{Insert} and \op{Delete} return without modifying the \textit{node tree}.
Their arrival point at $\ell$ is the last configuration during their execution in which $\ell$ is on the search path for $k$. Such a configuration exists by \Cref{lem:reachability}.

\end{enumerate}

We define inductively the arrival points of update cases that modify the \textit{node tree}. Assume the arrival points are defined for a prefix of the execution $\alpha$. Let $s$ be the next step that modifies the \textit{node tree}. The possible cases are as follows. 

\begin{enumerate}[resume]
\item\label{ap-insert-new}
Consider an \op{Insert}($k$) that executes a successful \op{SCX} step $s$ to replace a leaf $\ell$ by an internal Node \x{new} with two leaf children, \x{newLeaf} and $\ell'$, that contain $k$ and $\ell$'s key, respectively. 
This \op{SCX} is the arrival point at \x{new} and either the left or right child of \x{new} (depending on whether the operation's key is less than \x{new}.\x{key} or not) of
all operations whose arrival points at $\ell$ precede the SCX, in the order of their arrival points at $\ell$, followed by the \op{Insert}($k$) at both \x{newLeaf} and \x{new}.

\item\label{ap-delete}
Consider a \op{Delete}($k$) that performs a \x{SCX} step $s$ to modify the \textit{node tree}.  This step replaces an internal Node $p$ (whose children are a leaf $\ell$ containing $k$ and its sibling \x{sib}) by a new copy $\x{sib}'$ of \x{sib}.
For each operation on any key $k'$ whose arrival point at $\ell$ precedes $s$, $s$ is the operation's arrival point at $sib'$ and all of its descendants that are on the search path for $k'$.
Additionally, for each operation whose arrival point at $sib$ precedes $s$, $s$ is the operation's arrival point at $sib'$.
If multiple operations on $k'$ are assigned arrival points at the same node, they occur in the same order as their arrival points at $\ell$.
Finally, $s$ is also the arrival point of the \op{Delete}($k$) at $sib'$ and all its descendants that are on the search path for $k$.

\item\label{ap-rotate}
Consider a rebalancing operation that performs a successful SCX step $s$. Let $G_{old}$ be a patch of nodes in the \textit{node tree}, rooted at node $old$. Let $F$ be the set of nodes that are the children of nodes at the last level of the patch.
This step $s$ atomically modifies a \textit{node tree} by replacing the patch $G_{old}$, with a new patch $G_{new}$, rooted at a node $new$ and the same fringe $F$, in the \textit{node tree} (as shown in the rebalancing diagrams in~\cite[Figure 6.5]{Bthesis17}).

For each operation that arrived at a fringe node prior to $s$, $s$ serves as the operation's arrival point at every ancestor of that fringe node in $G_{new}$. The order of arrival points is same as it is in the fringe node. Operations from different fringe nodes are ordered according to their left-to-right position in the tree: operations from the left fringe node precede those from the right.
If the SCX replaces a leaf $\ell$ (if any) with new copy $\ell'$, then this SCX serves as the arrival point at $\ell'$ of all operations that had an arrival point at $\ell$ prior to the SCX, in the same order.


\item\label{ap-refresh}
Consider a successful \CAS\ performed by a \op{Refresh} $R$ on the \x{version} field of an internal Node $x$ at line \ref{lin:refreshCAS}.  
Let $x_L$ and $x_R$ be the childen of $x$ read by $R$ at line \ref{lin:readXl} or \ref{lin:readXr}.
The CAS is the arrival point at $x$ of 
\begin{enumerate}
\item\label[part]{a-left}
all operations
that have an arrival point at $x_L$ prior to $R$'s last read at line \ref{lin:readXlV1} or \ref{lin:readXlV2} and do not already have an arrival point at $x$ prior to the CAS, in the order of their arrival points at $x_L$, followed by
\item\label[part]{a-right}
all operations
that have an arrival point at $x_R$ prior to $R$'s last read at line \ref{lin:readXrV1} or \ref{lin:readXrV2} and do not already have an arrival point at $x$ prior to the CAS, in the order of their arrival points at $x_R$.
Note again that this preserves the order of operations on same key; the order of operations across the key need not be preserved and is irrelevant. 
\end{enumerate}
\end{enumerate}
\label{ap-definition}
\end{definition}

The arrival points at an arbitrary Node $x$ form a sequence of operations, referred to as \textit{Ops} sequence. Each element in this sequence is of the form $\ang{operation(k):response}$, where operation is either an insert or a delete with a boolean response attached. \textit{Ops} sequences track the operations that have arrived at a Node.

For deletes and inserts whose arrival point is defined by Part~\ref{ap-failed-delete}, the associated response is \textit{false}.
For inserts and deletes whose arrival points are defined by Part~\ref{ap-insert-new} and~\ref{ap-delete}, the associated response is \textit{true}. 
Finally, whenever arrival points are copied from a removed or replaced node to another, as in Part~\ref{ap-insert-new},~\ref{ap-delete},~\ref{ap-rotate} and~\ref{ap-refresh}, the associated responses are copied as well.

\begin{definition}
    For each configuration $C$ and Node $x$, 
    \begin{enumerate}
        \item Let \textit{Ops(C, x)} be the sequence of update operations with arrival points at $x$ that are at or before $C$, in order of their arrival points at $x$.
        \item Let \textit{Ops*(C, x)} be the sequence of update operations with arrival points at $x$ that are strictly before $C$, in the order of their arrival points at $x$.
        \item Let \textit{Ops(C, x, k)} be the subsequence of \textit{Ops(C, x)} consisting of operations with key $k$.
    \end{enumerate}
\end{definition}


\begin{observation}
    The CAS on line~\ref{lin:refreshCAS} never attempts to store a value in a node's version field that has previously been stored in it.  
    \label{obs:ver-allocs}
\end{observation}

\begin{lemma} 
If two calls to \op{Refresh} on the same Node perform successful \op{CAS} steps, then one performs the read on line~\ref{lin:readOldVer} after the CAS on line~\ref{lin:refreshCAS} of the other.
\label{lem:seq-cas}
\end{lemma}
\begin{proof}
    Without loss of generality, let \x{x} be an internal Node, and let $T1$ and $T2$ be two arbitrary processes.
    Suppose $T1$ reads the value $old$ from $x.version$ at line~\ref{lin:readOldVer} and successfully performs a CAS at line~\ref{lin:refreshCAS}, updating $x.version$ from \textit{old} to \textit{new}, where \textit{new} is a pointer to a Version object allocated at line~\ref{lin:allocateVersion}.     
    Now, assume that $T2$ reads old from $x.version$ immediately before $T1$'s CAS, and $T2$ also succeeds in its CAS at line~\ref{lin:refreshCAS}.
    
    For $T2$ to succeed, at the time of its CAS, it must see that $x.version$ is same as $old$. However, since $T2$ changes $x.version$ to new and from obeservation~\ref{obs:ver-allocs}, $x.version$ cannot be $old$, a contradiction. Therefore, $T2$'s read at line~\ref{lin:readOldVer} must be after $T1$'s successful CAS.
\end{proof}



The claims of the following Invariant ensure that no update is dropped by the propagation if a Node where it has arrived is removed from the \textit{node tree}.
In other words, propagation ensures two properties: the upward consistency of propagation, such that a $Ops$ sequence of a parent Node 
(on a particular key)
is always a prefix of its children's and descendants'; and the monotonic growth of all $Ops$ sequences of all nodes.

\op{Refresh} operation on a node \x{x} updates \x{x.version} using the versions of \x{x}'s children. Reading the version of a child (between lines~\ref{lin:Xlbegin}-\ref{lin:Xlend} and~\ref{lin:Xrbegin}-\ref{lin:Xrend}) is a three-step process:
\begin{enumerate}
    \item read pointer to \x{x}'s child \x{y};
    \item read the pointer to \x{y}'s version; and
    \item verify that \x{y} is still a child of \x{x}.
\end{enumerate} 
these steps repeat until the last step succeeds. Then following observation directly follows from the code of \op{Refresh}.

\begin{observation}
For a node \x{x} with child \x{y}, \op{Refresh} ensures that \x{y} was the child of \x{x} at the time it read \x{y}'s version. This implies that \op{Refresh} can add to the $Ops$ sequence of $x$ only those operations from $y$ that were added in $y$'s $Ops$ sequence before $y$ was replaced. 
    \label{obs:rd-chd-ver}
\end{observation}

This helps avoid the violation of~\invref{noAPLost} (part 1), where a concurrent Refresh on x uses child y's version after y was replaced by a new child y' such that y' contains operations of y from before it was deleted and since then y got a new version which was read by the Refresh at x. In this case, the~\invref{noAPLost} (part 1) is violated between x and y'. 

\begin{invariant}
For any configuration $C$, any key $k$, and any internal node $x$ with a child $y$ in $C$:
\begin{enumerate}
\item 
Every operation in $Ops(C, x, k)$ is also in $Ops(C, y, k)$.
\item 
If an \op{SCX} changes the child pointer of $x$ from $y$ to some node $y'$ at configuration $C'$, then every operation in $Ops(C, y, k)$ is also in $Ops(C\y{'}, y', k)$.
\end{enumerate}
\label{inv:noAPLost}
\end{invariant}

\begin{proof}
We will use induction on the sequence of configurations and argue that every possible modification to $Ops$ sequence of nodes preserves the claims.

The invariant holds vacuously in the initial configuration because no operations have arrival points. Assume that the claims hold up to some configuration $C$. We show that they hold up to the next configuration $C'$.
The following cases can occur:

\begin{enumerate}
    \item Part~\ref{ap-failed-delete} of Definition~\ref{ap-definition}.
    It adds the update operation only to the $Ops$ sequence of a leaf Node. Since a leaf does not have children, claim 1 is not violated. Additionally, since it does not modify the \textit{node tree}, claim 2 is trivially satisfied.
    
    
    \item Part~\ref{ap-insert-new} of Definition~\ref{ap-definition}.
    It ensures that each operation $op$ that have an arrival point at the replaced leaf $\ell$ is moved to the $Ops$ sequences of the two new leaves.  
    $op$ with key $k$ $<$ $new.key$ is moved to the left leaf and the $op$ with $k$ $\geq$ $new.key$ is moved to the right leaf. The insert itself is appended to the $Ops$ sequence of the appropriate leaf based on whether its key is less than or greater than equal to $new.key$.    
    Additionally, $op$ including the current insert is also added to the $Ops$ sequence of $new$.
    Therefore, claim 1 is preserved. 
    Moreover, no violation of the invariant is created at the node whose child pointer is changed to point to the new internal node $new$, since all arrival points of the replaced leaf are transferred to $new$. 
    Also, since every operation in the $Ops$ sequence of the old leaf $\ell$ is moved to $new$ and its appropriate leaf child Claim 2 is also preserved.

    \item Part~\ref{ap-delete} of Definition~\ref{ap-definition}. This has two cases. 

    First, if $sib'$ is a leaf, then satisfying Claim 1, is trivial as a leaf has no children. 
    Moreover, since all arrival points of the replaced leaf nodes ($\ell$ and $sib$) are transferred to the newly created leaf $sib'$, no violation of the invariant is introduced at the Node whose child pointer is updated to point to $sib'$.
    Every operation in $p$ is in either $\ell$ or $sib$, by induction hypothesis of claim 1.
    $sib'$ has all operations from the $Ops$ sequence of $\ell$ and $sib$ and also includes the current \op{delete}($k$) operation. Hence, $sib'$ has strictly larger $Ops$ sequence than $p$, implying every operation in $Ops$ sequence of $p$ is also in $sib'$ preserving Claim 2.

    Second, if $sib'$ is an internal Node, then the delete operation ensures that all operations that arrived at the removed leaf $\ell$ are transferred to the $Ops$ sequence of $sib'$ and to all its descendants along the search path of those operations. Additionally, all operations from $sib$ are moved to the $Ops$ sequence of $sib'$. By induction hypothesis all operation at $sib$ should already be there in descendants of $sib'$ as they do not change.
    This ensures that Claim 1 is preserved at $sib'$.
    
    Moreover, all arrival points of the replaced Nodes are transferred to $sib'$ and to all its descendants along the corresponding search paths. 
    Therefore, no violation of the invariant is introduced at the Node whose child pointer is updated from $p$ to $sib'$. In fact, as explained in the first case above, $Ops$ sequence of $sib'$ is strictly larger than the $Ops$ sequence of removed $p$. Thus, Claim 2 is preserved.

    
    \item Part~\ref{ap-rotate} of Definition~\ref{ap-definition}.
    A rebalance operation ensures that all operations whose arrival points are at fringe Nodes of the replaced sub graph are added to the $Ops$ sequences of the newly created ancestors of those fringe Nodes.
    If the replaced subgraph contains leaves (leaves do not have any fringe Nodes), then all operations with a key $k'$ that have their arrival points at those leaves are added to the newly created leaf copies and their appropriate ancestors in the search path of $k'$.
    Meaning that every operation in the $Ops$ sequence of a newly created parent in the new patch also appears in the $Ops$ sequence its children (if they exist).    
    Thus, Claim 1 is preserved. 
    which lazily propagate the operations to the $Ops$ sequence of the new ancestors when their version pointers become non-nil.
    
    Similarly, no violation of the invariant is introduced at the Node whose child pointer is updated, since all arrival points at the replaced child $old$ are transferred to the $new$ node that replaces it. As a result, all operations in the $Ops$ sequence $old$ are incorporated into the $Ops$ sequence of $new$. 
    Consequently, Claim 2 is also preserved.
    \item 
Part~\ref{ap-refresh} of Definition~\ref{ap-definition}.
    Consider new operations added to the $Ops$ sequence of Node $x$ by a \op{Refresh} when it changes $x.version$ (at \Lineref{refreshCAS}). 
    There are two cases.
    
    First, suppose the children of $x$ have not changed since their versions were read before $C$. Then, these new operations were in the $Ops$ sequence of $x$'s children last time their (non-nil) version pointers were read.
    Thus, Claim 1 is preserved at $x$ in a later configuration $C'$.

Second, suppose the children of $x$ have changed since their versions were read before $C$. 
There can be two possibilities within this case. Either new operations were not added to the $Ops$ sequences of the replaced children before $C$ or they were added. In both cases,
by \obsref{rd-chd-ver}, \op{Refresh} propagates only those operations $O$ that were in $Ops$ sequences of children before they were replaced.
Additionally, by claim 2 of the induction hypothesis, all such operations $O$ are also present in the $Ops$ sequences of the new children. Therefore, in $C'$, Claim 1 is preserved between $x$ and its new children.

\end{enumerate}
\end{proof}

\begin{observation}
For any leaf, version pointer field is never nil.
\end{observation}
This directly follows from Definition~\ref{ap-definition}.

Lemma~\ref{lem:sref-propx} and Lemma~\ref{lem:dref-propx} are identical to Lemmas 21 and 22 of FR-BST~\cite{FR24}.   

The following lemma shows that a successful \op{Refresh} operation on a node propagates all operations that arrived at its children before the refresh read their version pointers.
\begin{lemma}
    Suppose Node $y$ is a child of Node $x$ at a configuration $C$ and that an update operation $op$ has an arrival point at $y$ at or before $C$. 
    If \op{Refresh(x)} reads $x.version$ at line~\ref{lin:readOldVer} after $C$ and performs a successful CAS on line~\ref{lin:refreshCAS} then op has an arrival point at $x$ at or before the CAS.
    \label{lem:sref-propx}
\end{lemma}
\begin{proof}
    Assume $y$ is the left child of $x$ in $C$. (The case where $y$ is the right child of $x$ is symmetric.)
    
    To ensure that $op$ is propagated to $x$, we consider the following cases:
    \begin{enumerate}
        \item 
        If the \op{Refresh} reads $y$ as the left child of $x$ at line~\ref{lin:readXl} after $C$, then $op$ already has arrival point at $y$, by the assumption of this lemma.
        Now, when $op$ reads $y.version$ at line~\ref{lin:readXlV1}, 
        by Definition~\ref{ap-definition}, part~\ref{a-left}, $op$ has an arrival point at $x$ at or before the CAS at line~\ref{lin:refreshCAS}.        
        \item 
        If the \op{Refresh} reads a different Node $y'$ as the left child of $x$ at line~\ref{lin:readXl} after $C$, then by the second claim of Invariant~\ref{inv:noAPLost}, op has an arrival point at $y'$ before this read occurs.
    \end{enumerate}
    In either case, op has an arrival point at the left child of $x$ no later than line~\ref{lin:readXlV1} or line~\ref{lin:readXlV2} (whichever occurs last), and strictly before the successful CAS performed by \op{Refresh} at line~\ref{lin:refreshCAS}. Therefore, by Definition~\ref{a-left} and recursive \op{Refresh} mechanism, $op$ has an arrival point at $x$ at or before the successful CAS of the \op{Refresh}.  
\end{proof}

\begin{lemma}
    Suppose Node y is a child of Node x at a configuration C and that an update operation op has an arrival point at y before C. If a process executes the double refresh at lines~\ref{lin:r1}-\ref{lin:r2} on x after C then op has an arrival point at x at or before the end of the double refresh.
    \label{lem:dref-propx}
\end{lemma}
\begin{proof}
    If either call to \op{Refresh} successfully performs the CAS step at line~\ref{lin:refreshCAS}, then the claim directly follows from Lemma \ref{lem:sref-propx}.
    
    Now consider the case where both \op{Refresh} operations, $R_1$ and $R_2$, fail their CAS steps. This can only occur if two other \op{Refresh} operations performed successful CAS steps, $c_1$ and $c_2$, during the execution of $R_1$ and $R_2$, respectively.

    Let $R$ be the \op{Refresh} operation that executes the CAS step $c_2$.
    By Lemma~\ref{lem:seq-cas}, the read at line~\ref{lin:readOldVer} in $R$ must occur after the successful CAS step $c_1$, implying that $R$ started after configuration $C$. 
    Now, if $op$ has an arrival point at $y$ in $C$, by Invariant~\ref{inv:noAPLost}, part 2, $op$ still has arrival point at a child of $x$, when R2 reads a pointer to its child.
    
    Applying Lemma~\ref{lem:sref-propx} to $R$ implies that op has an arrival point at $x$ no later than $c_2$, which is before the end of $R_2$.
\end{proof}

\begin{lemma}
    If $op$ has arrived at the leaf at the end of the search path for a key $k$ in the \textit{node tree} $T_c$ of configuration $C$, then in any configuration $C'$ later than $C$, $op$ has arrived at the leaf at the end of the search path for $k$ in $T_c'$.
    \label{lem:ap-hz-move-leaf}
\end{lemma}
This directly follows from Invariant~\ref{inv:noAPLost}, part 2.

\begin{lemma}
    If $y$ is a child of $x$ in some configuration $C$ and $y'$ is the child of $x$ in some later configuration $C'$, the $Ops(C, y, k)$ is a prefix of $Ops(C', y', k)$.
    \label{lem:ap-hz-move-node}
\end{lemma}
This directly follows from Invariant~\ref{inv:noAPLost}, part 2.


In this section, we consider an update operation $op$ on a key $k$ and show that it has an arrival point at the root before it terminates.
Let $C_0$ be the first configuration at or before op calls \op{Propagate}.
Let $x_1, \cdots, x_m$ be the Nodes on the local stack (from the newest pushed to the oldest) in $C_0$.
Since $x_1$ is on the stack, it is an internal Node. Let $x_0$ be the left child of $x_1$ in $C_0$ if $k<x_1.key$ or the right child of $x_1$ otherwise.
We show by induction on $i$ that before $op$ adds $x_i$ to its refresh set, $op$ has an arrival point at $x_i$.
We first prove the base case by showing that $op$ has an arrival point at $x_0$ before \op{Propagate} is called. Then move on to show that when $op$ adds some node $x_i$ from $x_1, \cdots, x_m$ to its refresh set, in some later configuration, then the $op$ has an arrival point at x and at all its descendants in its search path.

\begin{lemma}
op has an arrival point at $x_0$ at or before $C_0$, where $C_0$ is a configuration before op invokes \op{Propagate}.
\label{lem:base-prop-root}
\end{lemma}
\begin{proof}
We consider several cases.
\begin{itemize}
    \item 
    Suppose $op$ is either a delete or an insert operation, as described in part~\ref{ap-failed-delete} of Definition~\ref{ap-definition}.
    Then $op$ reaches a leaf Node $\ell$ from some internal node $x$ and determines that $\ell$ does not contain key $k$. 
    The arrival point of $op$ at $\ell$ is the latest configuration $C$ prior to $C_0$ in which $op$ reads a pointer to $\ell$.
    By Lemma~\ref{lem:ap-hz-move-leaf}, in the later configuration $C_0$, the arrival point of $op$ is at the leaf $x0$ in the search path of $k$.

    \item 
    Suppose $op$ is an update operation, as described in part~\ref{ap-insert-new} of Definition~\ref{ap-definition}. By definition, the SCX that adds a new leaves Node $\ell$ and $newLeaf$ is the arrival point of the $op$ at $newLeaf$. Let $C$ be the configuration immediately after this SCX. 
    Then $\ell'$ lies on the search path for key $k$ in configuration $C$.
    
    To see this, let $x$ be the node whose child pointer is changed by the SCX to add $\ell'$.
    By Lemma~\ref{lem:reachability}, $x$ was reachable in some earlier configuration prior to the SCX. By Lemma~\ref{lem:lem13}, $x$ remains reachable at the time the SCX is performed. As a result, By Lemma~\ref{lem:lem14}, $x$ lies on the search path for $k$ when SCX occurs.
    Therefore, $\ell'$ is on the search path for $k$ in configuration $C$.
    
    Now, consider a later configuration $C_0$ in which the search path for $k$ ends at a (possibly different) leaf Node $x_0$.
    Then, by Lemma~\ref{lem:ap-hz-move-leaf}, the arrival point of $op$ lies at $x_0$, the leaf in the search path for $k$ at $C_0$, whether or not $x_0 = \ell'$.

    \item If $op$ is a \op{delete}($k$) operation, as described in part~\ref{ap-delete} of Definition~\ref{ap-definition}, then at the time SCX modifies the tree, $op$ has an arrival point at a descendant leaf of $sib'$ that lies on the search path of $k$. Further, any subsequent modification to the tree still ensures that $op$ has an arrival point at $x_0$ before $C_0$, by arguments similar to those give above. 

    \item Similarly, if $op$ is a rebalance operation, as described in part~\ref{ap-rotate} of Definition~\ref{ap-definition}. Then the same reasoning as above applies to show that the arrival point of $op$ is at the leaf $x_0$ in configuration $C_0$.
\end{itemize}
\end{proof}


\begin{lemma}
op has an arrival point at root before it terminates.
\end{lemma}
\begin{proof}
We prove by the induction on the sequence of nodes, $x_1, \cdots, x_m$, $op$ adds to its refreshed set during execution of \op{Propagate}.

    \textbf{Base case:} $op$ has an arrival point at the leaf $x_0$ in the search path of op's key $k$. This follows from lemma~\ref{lem:base-prop-root}.
    
    \textbf{Induction step:} Suppose op adds Node $x_i$ to its refreshed set and has an arrival point at the Node. We need to show that before op adds a node $x_{i+1}$ to its refreshed set it will have an arrival point at $x_{i+1}$. 
    Note, in order to add $x_{i+1}$ to its refreshed set, op first reads $x_{i+1}$ from top of the stack, reads a child $x_j$ of $x_{i+1}$. \textit{W.l.o.g.,} assume $x_j$ is the left child of $x_{i+1}$ such that $x_j.key$ $<$ $x_{i+1}.key$. (The argument when $x_j$ is the right child, i.e., $x_j.key$ $\geq$ $x_{i+1}.key$ is symmetric.) Note that by the way Propagation adds Nodes in its refresh set, there are following cases for $x_j$:

    \textbf{Case 1:} Node $x_j$ $=$ $x_i$. Then $x_j$ could either be a leaf or an internal node in the refreshed set. In the former case, op will execute a double refresh before adding $x_{i+1}$ to its refresh set followed by removing $x_{i+1}$ from the stack. This means from lemma~\ref{lem:dref-propx}, op will also have its arrival point at $x_{i+1}$, which is before the double refresh ends and thus is definitely before $x_{i+1}$ is removed from the stack.
    
    In the latter case, $x_i$ is the left child of $x_{i+1}$ already in the refreshed set of op. Therefore, op will execute a double refresh on $x_{i+1}$ ensuring it arrives at $x_{i+1}$ before the double refresh returns and then removes $x_{i+1}$ from the stack followed by adding $x_{i+1}$ to its refresh set. Therefore, the op has an arrival point at $x_{i+1}$ before it adds $x_{i+1}$ to its refresh set and therefore definitely before propagation terminates.

    \textbf{Case 2:} Node $x_j$ $\neq$ $x_i$. Then, $x_j$ could either be a leaf or the internal node \textbf{not} in the refreshed set of op. In the former case, op will have its arrival point moved to the leaf $x_j$ from lemma~\ref{lem:ap-hz-move-leaf}. Therefore, by lemma~\ref{lem:dref-propx}, before $op$ returns from the double refresh at $x_{i+1}$, $op$ has its arrival point at $x_{i+1}$. 
    
 In the latter case, $x_{i+1}$ is not the next node to be added to the refreshed set, contradicting the choice of $x_{i+1}$.
    Infact, $op$ will traverse down the new search path for $k$ adding the nodes it encounters to its stack until the traversal hits a leaf node or an internal node already in the refresh set (refer to \op{Propagate()} in Figure~\ref{abt-pseudocode}).
    In either case, the inductive argument restarts and follows the same reasoning as in Case 1.
    
    Since \op{Propagate} procedure only terminates when the stack is empty, this implies that the root must have been added to the refresh set and subsequently removed from the stack.
    The double refresh at the root ensures that $op$ has an arrival point there prior to its removal from the stack.
    So, it follows by induction that $op$ has an arrival point at the root before $op$ terminates.    
\end{proof}

\subsection{Linearization is consistent with responses}

The following invariant ensures that the $Ops$ sequences associated with nodes in a Node $x$'s left subtree only contain operations on keys less than $x.key$. Similarly, $Ops$ sequences for nodes in the right subtree only contain operations on keys greater than or equal to $x.key$. This is crucial for maintaining the BST property in the Version tree, even as the \textit{node tree} undergoes structural changes.

\begin{invariant} (BST Property of Ops.)
    For every internal Node x with any Node $x_L$ and $x_R$ in its left or right subtree, respectively, in any configuration C,
    \begin{enumerate}
        \item every operation in $Ops(C, x_L)$ has key less than x.key 
        \item every operation in $Ops(C, x_R)$ has key greater than or equal to x.key 
    \end{enumerate}
    \label{inv:subtree_opskeys}
\end{invariant}
\begin{proof}
    Initially, all $Ops$ sequences are empty, so the both claims are trivially true.
    
    Assume both claims hold up to some configuration $C$. We show that both claims will hold in a later configuration $C'$. Let $s$ be a step that adds or removes a key from the left subtree of Node $x$, leading to the configuration $C'$. The argument for the right subtree is symmetric.
    We consider various types of steps that can add or remove keys, or modify $Ops$ sequences within the left subtree of $x$.   
       
   \textbf{Case 1:} $s$ is a successful SCX by a $delete(k)$ operation described in part~\ref{ap-failed-delete} of Definition~\ref{ap-definition} (The argument for the insert operation is similar).   
   In this case, the Operation is added to the $Ops$ sequence of $\ell$ at the time when it read the left child pointer of $x$. Since, the tree is not modified, $\ell$ remains reachable by $x$'s left child pointer, and $l.key$ $=$ $k$ $<$ $x.key$, preserving the Invariant. 
   
   by Definition, when $s$ replaces $\ell$ by $\ell'$ all operations in the $Ops$ sequence of $\ell$ are transferred to $\ell'$. Such that $Ops(C', \ell')$ $=$ $Ops(C', \ell') \langle insert(k): false\rangle$. By lemma~\ref{lem:staysCT}, the chromatic tree is still a BST after $s$. Since $\ell$ and $\ell'$ are left child of $x$, all keys in $Ops$ sequence of $\ell'$ are strictly less that $x.key$.

    
    \textbf{Case 2:} 
    $s$ is a SCX due to $insert(k)$ described in part~\ref{ap-insert-new} of Definition~\ref{ap-definition}.
    
    In this case, $s$ replaces a leaf $\ell$ by a new internal Node $new$ with two new leaf children, $newLeaf$ and $\ell'$.
    $newLeaf$ contains $k$, and $\ell'$ contains $l.key$

    Lemma~\ref{lem:staysCT} implies that the \textit{node tree} is a chromatic tree in all configurations and therefore always satisfies the BST property. This implies that when $new$ is added to left subtree of $x$ in $C'$, then $new.key$ $<$ $x.key$.
    
    By Definition~\ref{ap-definition}, part~\ref{ap-insert-new}, all previously arrived operations (including the current insert) with keys less than $new.key$ are transferred to $new$'s left child, and those with keys greater than or equal to $new.key$ are transferred to $new$'s right child.

    As a result, all operations in the $Ops$ sequence of $new$ will have keys less than $x.key$. Further, within the subtree at $new$, the operations in $Ops$ sequence of $\ell'$ and $newLeaf$ are distributed according to $new.key$, preserving the invariant.
    

    

    \textbf{Case 3:} $s$ is a SCX for $delete(k)$ described in part~\ref{ap-delete} of Definition~\ref{ap-definition}.
    
    In this case, in the subtree of $x$, step $s$ replaces an internal Node $p$ with key $k$ (whose children are a leaf $\ell$ with key $k'$ and its sibling $sib$ with key $k''$) by a new copy $sib'$ of $sib$. 
    By lemma~\ref{lem:staysCT}, the resulting tree in configuration $C'$ after the step $s$ remains a chromatic tree. Therefore all keys for the subtree rooted at the new internal node $sib'$ are less than $x.key$. 
    
    Moreover, from part~\ref{ap-delete} of Definition~\ref{ap-definition}, this step $s$ is the arrival point of all operations that arrived at $\ell$ and $p$ before $s$ at $sib'$ and at all its descendants who are in the search path of these operations. 
    As a result, the $Ops$ sequence of $sib'$ is updated to contain all operations who arrived at $\ell$, $p$ or $sib$ before step $s$ with keys less than $x.key$. Additionally, all descendants of $sib'$ whose $Ops$ sequence is updated is appended with all operations which arrived at $\ell$ or $p$ before $s$, such that for any descendant $d$ with $d_L$ and $d_R$ as its left and right child, in configuration $C'$, the $Ops$ sequence at $d_L$ contains all operations with keys less than $d.key$ and the $Ops$ sequence at $d_R$ contains all operations with keys greater than equal to $d.key$.      
    Consequently, in configuration $C'$, the updated ops sequences at $sib'$ and all its descendants follow claim 1 and claim 2 and all operations in any $x_L$ in the left subtree of Node $x$ have keys less than $x.key$. 

   \textbf{Case 4:} $s$ is a successful SCX for rebalancing operation described in part~\ref{ap-rotate} of Definition~\ref{ap-definition}. 
   It atomically replaces a subgraph of Nodes $G_{old}$ with a new subgraph $G_{new}$, where $old$ and $new$ are roots of their respective subgraphs and are the left child of $x$.   
   By lemma~\ref{lem:staysCT}, the resulting subtree remains a chromatic tree in $C'$, infact, $s$ does not add any new key and maintains the overall BST structure. Only the $Ops$ sequences of the new Nodes in $G_{new}$ are updated.
   By Definition~\ref{ap-definition}, part~\ref{ap-rotate}, operation from fringe Nodes of $G_{old}$, which is same as the fringe Nodes in $G_{new}$, are transferred to the $Ops$ sequences of appropriate new ancestors in the search path of their keys. Therefore, all keys in $Ops$ sequences of operation in the left subtree of $x$ remain less than $x.key$ as they were before the rebalance step. 
   
   \textbf{Case 5:} $s$ is a successful CAS by a refresh operation described in part~\ref{ap-refresh} of Definition~\ref{ap-definition}. 
   Consider that $s$ changes a version field of a Node $y$ in the left subtree of $x$. Let, $y_L$ and $y_R$ be the left and right child of $y$ read by its \op{Refresh}, respectively.
   Since lemma~\ref{lem:staysCT} implies that in all configurations BST property is preserved, at the time $s$ occurs at $y$, keys of all operations added to the $Ops$ sequence of $y$ in $C'$ are strictly less than $x.key$. Precisely, by Definition~\ref{ap-definition}, part~\ref{ap-refresh}, operations added to $Ops(C', y)$ come from the $Ops$ sequence of $y_L$ and $y_R$.


    
\end{proof}

\begin{observation}
    For each leaf Node x, x.version always points to a Version with key x.key and no children.
\end{observation}



\begin{invariant} (prefix property of Ops.)
    For all configurations and all Nodes x that are reachable in C, if $x_L$ and $x_R$ are the left and right child of x in C, then 
    \begin{enumerate}
        \item for keys k $<$ x.key, Ops(C, x, k) is a prefix of Ops(C, $x_L$, k), and
        \item for keys k $\geq$ x.key, Ops(C, x, k) is a prefix of Ops(C, $x_R$, k).
    \end{enumerate}
    \label{inv:opsk_prfx}
\end{invariant}
\begin{proof}
    In the initial configuration, $Ops$ sequences of all Nodes are empty so the claim is trivially satisfied.

    Assume that the claim is satisfied in some configuration $C$. We want to show that after a step $s$ that changes $Ops$ sequence of a Node $x$ leading to a configuration $C'$, the claim still holds.
    The step $s$ could occur due to an insert, delete, rebalance and refresh. we consider them one by one and argue that the invariants holds after a step $s$.

    First, consider a delete case that reaches a leaf $\ell$ and return $false$ without modifying the \textit{node tree}, as described in part~\ref{ap-failed-delete} in Definition~\ref{ap-definition}. Let $\ell$ be the left child of its parent $p$ (case when $\ell$ is a right child is symmetric). 
    Here in configuration $C$ since delete reached $\ell$, it is reachable from $p$, by lemma~\ref{lem:reachability}. Additionally, since the \textit{node tree} is a BST in all configurations, key $k'$ for delete is less than $p.key$ (Lemma~\ref{lem:staysCT}).
    In configuration $C$, by induction hypothesis, for all keys $k$ $<$ $p.key$, $Ops(C, p, k)$ is a prefix of $Ops(C, l, k)$.
    This delete adds its arrival point for $k'$ $<$ $p.key$ to $\ell$. Such that, $Ops(C', l, k)$ $=$ $Ops(C, l, k) \langle delete(k):false\rangle$. Thus, $Ops$ sequence of $\ell$ has only grown, ensuring in $C'$, $Ops(C', p, k)$ is a prefix of $Ops(C', l, k)$. 

    The insert case that returns false is similar.
        

    Consider the $insert(k)$ that replaces a leaf $\ell$ by a new internal Node $new$ with two new leaf children $\ell'$ and $newLeaf$, as described in part~\ref{ap-insert-new} in Definition~\ref{ap-definition}. Let $new$ be the left child of $p$ and $newLeaf$ be the left child of $new$. The arguments for other cases are similar.

    First we will tackle the relation between $new$ and the two new leaves.
    In $C'$, all the arrival point at $\ell$ are transferred to $new$ and $newLeaf$, plus the arrival point of insert.
    Therefore, by Definition~\ref{ap-definition}, part~\ref{ap-insert-new}, in $C'$, for all keys keys $k'$, $k'$ $<$ $new.key$, $Ops(C', new, k')$ is a prefix of $Ops(C', newLeaf, k')$. The key $k'$, also includes $k$.     
    Similarly, in $C'$, for all keys keys $k''$ $\geq$ $new.key$, $Ops(C', new, k'')$ is a prefix of $Ops(C', \ell', k'')$. Note that $k''$ $\neq$ $k$. Therefore, $Ops(C', \ell', k)$ is empty.  

    Now, consider the relation between $new$, $p$ and $\ell$.   
    In $C$, for all keys $k'$ $<$ $p.key$ ($k'$ also includes $k$), $Ops(C, p, k')$ is a prefix of $Ops(C, l, k')$ (By induction hypothesis).

    In $C'$, note for a delete that returns fail, its arrival point is at $C'$ and not $s$.
    Therefore, for all keys $k'$ $\neq$ $k$ and $k'$ $<$ $p.key$, $|Ops(C', new, k')|$ $\geq$ $|Ops(C, l, k')|$. For $k'$ $=$ $k$, $|Ops(C', new, k')|$ $=$ $|Ops(C, l, k')|$. In sum, the $Ops$ sequence at $new$ can only grow, combined with induction hypothesis, for all keys $k'$ $<$ $p.key$,  $Ops(C', p, k')$ is a prefix of $Ops(C', new, k')$. 
    Thus, the the invariant is preserved.
    

    Consider the delete that replaces an internal Node $p$ whose children are a leaf $\ell$ and its sibling $sib$ by a new node $sib'$, as described in part~\ref{ap-delete} in Definition~\ref{ap-definition}.
    Assume $p$ is left child of $gp$ and $\ell$ is left child of $p$ in $C$, therefore, $sib'$ is the left child of $gp$ in $C'$. The other case is symmetric.
    
    The delete transfers all prior arrival point from $p$, $\ell$ and $sib$ to $sib'$ and to the descendants of $sib'$. 
    Precisely, the prior arrival point for a key $k$ in $C$ is transferred to all descendant nodes that are in the search path for $k$ in $C'$. In addition, the delete also is assigned an arrival point at $sib'$ and all the descendants in its search path in $C'$. 
    As a result, every op on a key $k$ $<$ $gp.key$ in the $Ops(C, p, k) \cup Ops(C, l, k) \cup Ops(C, sib, k)$ is appended to Nodes in the search path for k in the subtree rooted at $sib'$ in $C'$. 
    Let, $d$ be an arbitrary node in the search path for $k$ and $dp$ be its parent, then for each such $d$ and $dp$ pair, $Ops(C', dp, k)$ is a prefix of $Ops(C', d, k)$. By lemma~\ref{lem:staysCT}, such a search path exists in the subtree rooted at $sib'$ and nowhere else.

    What remains for the delete case is showing that the invariants are preserved for $sib'$ and $gp$ in $C'$.
    By part~\ref{ap-delete} of Definition~\ref{ap-definition},
    for all keys $k$ < $gp.key$, $|Ops(C', sib', k)|$ $\geq$ $|Ops(C, sib, k)| + |Ops(C, l, k)|$. Note for the delete returning false, the arrival point can be at $C'$, hence, at $C'$, $sib'$ can have more operations that those transferred from leaves in $C$.

    By induction hypothesis,
    $Ops(C, gp, k)$ is a prefix of 
    
    \noindent
    $Ops(C, p, k)$ which is a prefix of $Ops(C, sib, k) \cup Ops(C, l, k)$. Not that the $Ops$ sequence of $Ops(C', sib', k)$ has only grown, such that $|Ops(C', sib', k)|$ $\geq$ $|Ops(C', gp, k)|$. Consequently, $Ops(C', gp, k)$ is a prefix of $Ops(C', sib', k)$.



    Consider a rebalance operation RB1 in~\cite[Chap 6, Fig. 6.5]{Bthesis17}. All arrival points from the fringe nodes in configuration $C$ are transferred to the newly created nodes in the configuration $C'$, by part~\ref{ap-rotate} in Definition~\ref{ap-definition}. Let $n$ and $n_R$ be the two newly added nodes, such that $n_R$ is right child of $n$. 

\begin{itemize}
    \item 
    For all keys $k$ $\geq$ $n_R.key$ $\geq$ $n.key$, 
    $Ops(C', u_{xr}, k)$ $=$ $Ops(C', n_R, k)$ $=$ $Ops(C', n, k)$.

    \item For all keys $k$, such that $n_R.key$ $>$ $k$ $\geq$ $n.key$,
    
    $Ops(C', u_{xlr}, k)$ $=$ $Ops(C', n_R, k)$ $=$ $Ops(C', n, k)$.

    \item For all keys $k$ $<$ $n.key$, 
    $Ops(C', u_{xll}, k)$ $=$ $Ops(C', n, k)$.
\end{itemize}
    Let, $u_x$ be the left child of $u$, since in configuration $C$, for all keys $k$ $<$ $u.key$, $Ops(C, u, k)$ is a prefix $Ops(C, u_{xll}, k) \cup Ops(C, u_{xlr}, k) \cup Ops(C, u_{xr}, k)$. Further in $C'$, $Ops$ sequence of $n$ can only grow, implying $Ops(C', u, k)$ is a prefix of $Ops(C', n, k)$.     
    
It remains to check that a successful CAS step of a \op{Refresh(x)} that updates $x.version$ preserves the invariant. 
Parts~\ref{ap-refresh} of Definition~\ref{ap-definition} appends new operations to $Ops(C', x, k)$.
By Invariant~\ref{inv:subtree_opskeys}, the new operations can only come from $x_L$ if $k$ $<$ $x.key$ or $x_R$ if $k$ $\geq$ $x.key$. Thus, in case $x_L$ and $x_R$ were still the child in $C'$, $Ops(C', x, k)$ simply becomes a longer prefix of $Ops(C', x_L, k)$ or $Ops(C', x_R, k)$.

Suppose, $x$'s children changed from $x_L$ and $x_R$ to $x_L'$ and $x_R'$, respectively, at some point before the CAS in \op{Refresh} but after the versions of $x_L$ and $x_R$ were read.
Even in this case, the invariant holds because, by Claim 2 of \invref{noAPLost}, every operation that arrived in the $Ops$ sequences of $x_L$ and $x_R$ before they were replaced is transferred to $x_L'$ and $x_R'$. Further, by \obsref{rd-chd-ver}, only those operations that arrived at $x_L$ and $x_R$ before they were replaced by the new children are added to the $Ops$ sequence of $x$.

    
    \end{proof}

\begin{lemma}
    In every configuration C, for every Node x that is reachable in C and has a non-nil version, 
    \begin{enumerate}
        \item the Version tree rooted at x.version is a BST whose leaves contains exactly the keys that would be in a set after performing Ops(C, x) sequentially, and 
        \item the responses recorded in Ops(C, x) are consistent with  executing the operations sequentially.
    \end{enumerate}
    \label{inv:version_keys_resp}
\end{lemma}

\begin{proof}
In the initial configuration, both claims hold trivially because each Node’s version tree is empty, and thus all corresponding $Ops$ sequences are empty.

Lets consider a step $s$ of an update operation that either change the number of keys or the $Ops$ sequence of a Node. The step $s$ leads from a configuration $C$ to a configuration $C'$. Lets assume that both the claims hold up to $C$ and we will prove that both the claims will still hold in $C'$ that results from the execution of step $s$. The step $s$ could be due to several cases and we will show by induction that in each of the case the claims hold.

\textbf{Case 1:} Let $s$ be due to either $delete(k)$ or $insert(k)$ in part~\ref{ap-failed-delete} of Definition~\ref{ap-definition}. In this case, $s$ modifies the $Ops$ sequence of leaf $\ell$ with some key $k'$. Specifically, 

$Ops(C', l)$ $=$ $Ops(C, l). \langle delete(k), false\rangle$ and

$Ops(C', l)$ $=$ $Ops(C, l). \langle insert(k), false\rangle$, for the delete and insert operation, respectively. 
Since both operations return false, it does not change the number of keys stored at $\ell$ in $C$ and $C'$.
In both the configurations the set remains $\{k'\}$. So the claims hold.



\textbf{Case 2:} Let $s$ be due to the insert on key $k$ in part~\ref{ap-insert-new} of Definition~\ref{ap-definition}.
In this case, $s$ replaces $\ell$ having key $k'$ by a Node $new$ whose two children are $\ell'$ and $newLeaf$. $W.l.o.g$, assume $\ell'$ is the left child of $new$ and $newLeaf$ is the right child. Then, $new$ and $newLeaf$ are assigned the key $k$ and $\ell'$ gets the key $k'$.

For all keys $k''$ $<$ $k$, $Ops(C', \ell', k'')$ $\geq$ $Ops(C, l, k'')$. Note that delete ops returning false can have arrival points at $C'$ and not $s$.
The $Ops$ sequence at $\ell'$, for $k''$ $=$ $k'$ can either result in $\ang{Ops(insert(k''):true)}$ or $\ang{Ops(insert(k''):false)}$. For $k''$ $\neq$ $k'$, it will end in $\ang{Ops(delete(k''):true)}$
Since the claim holds for $Ops(C, l, k'')$ and since the set of keys do not change in $Ops(C', l, k'')$, the claim holds in $C'$, where the set of keys in $C$ and $C'$ is  $\{k'\}$.

For all keys $k'''$ $\ge$ $k$, $Ops(C', newLeaf, k''')$ $\geq$ $Ops(C, l, k''')\cdot \langle insert(k): true\rangle$.
For $k'''$ $\neq$ $k$, performing $Ops(C, newLeaf, k''')$ will yield same set as performing $Ops(C', newLeaf, k''')$, by induction hypothesis. 
The corresponding $Ops(C', newLeaf, k''')$ will end in $\ang{delete(k'''):false}$ or $\ang{delete(k'''):true}$.
Appending $\langle insert(k): true\rangle$ to $Ops(C', newLeaf, k''')$, will yield a set $\{k\}$, which would be same as if $Ops(C, l, k''')\cdot \langle insert(k): true\rangle$ is performed sequentially.

Remaining is proving that claim holds at $new$ as well. By the definition, at $new$, $Ops(C', new)$ $=$ $Ops(C, l)\cdot \langle insert(k): true\rangle$. 
The $Ops(C', new)$ for all keys $k''$ other than $k'$ and $k$, will end with either $\ang{delete(k''):false}$ or $\ang{delete(k''):true}$. 
For $k'$, $Ops(C', new)$ will either end $\ang{Ops(insert(k'):true)}$ or $\ang{Ops(insert(k'):false)}$
For $k$, $Ops(C', new)$ will end in $\ang{Ops(insert(k):true)}$.
Performing $Ops(C', new)$ sequentially, will yield a set $\{k', k\}$. 
Thus the claims hold.

\textbf{Case 3:} Let $s$ be due to the delete on key $k$ in part\ref{ap-delete} of Definition~\ref{ap-definition}.
In this case, $s$ replaces $p$, which has two children, $\ell$ with key $k$ and $sib$ with key $k'$, with $sib'$ having key $k'$. 
\textit{W.l.o.g.}, assume $\ell$ is left child of $p$ and $p$ is left child of its parent (other cases can be argued similarly).

For any key $k''$, all operations in the $Ops$ sequences of the removed nodes $p$, $\ell$ and $sib$ in $C$  are moved to the $Ops$ sequence of $sib'$ and all its descendants up to a leaf in the search path of $k''$ in $C'$. 
Note that such a search path exists by Lemma~\ref{lem:staysCT}.
Let $z$ be any node in the search path for $k''$ $=$ $k$ from the sub tree rooted at $sib'$ in $C'$.
Then, the $Ops$ sequence of $z$ in $C'$, ends with $\ang{delete(k): true}$, 
For the node $z$ in the search path of $k''$ $\neq$ $k$, the $Ops(C', z, k'')$ either remains same as $Ops(C, z, k'')$ or ends with either $\ang{delete(k''): true}$ or $\ang{delete(k''): false}$.
Thus, for all keys, $k''$, $Ops(C', sib', k'')$ yields a set of keys which is same as $Ops(C, sib, k'')$. This reflects the set of keys is the leaves of the version tree at $sib'$.
Additionally, since the set of responses move along with the operations in the $Ops$ sequences claim 2 also holds.





\textbf{Case 4:} Let $s$ be due to the rebalancing in part~\ref{ap-rotate} of Definition~\ref{ap-definition}.
In this case, $s$ replaces a subtree rooted at a Node $old$ (a child of Node $u$) by a new subtree rooted at $new$. 
This may rearrange the keys of the replaced Nodes but does not change the set of keys in the subtree(see proof in ~\cite[Theorem 6.4]{Bthesis17}.
From Invariant~\ref{inv:subtree_opskeys}, the $Ops$ sequences of all the replaced nodes are transferred to newly created nodes in a way that for each newly created node $z$ with any Node $z_L$ and $z_R$ in its left or right subtree, all operations in $Ops(C', z_L)$ have key less than $z.key$ and all operations in $Ops(C', z_R)$. Moreover, from Invariant~\ref{inv:opsk_prfx}, for all keys $k$ less than $z.key$, $Ops(C', z, k)$ is a prefix of $Ops(C', z_L, k)$ and for all $k$ greater than or equal to $z.key$, $Ops(C', z, k)$ is a prefix of $Ops(C', z_R, k)$. By observation, rebalance step's net effect is that it changes the search paths of existing keys at the leaf level without changing the number of keys at the leaf level. Therefore, in $C'$, both the claims hold.


\textbf{Case 5:} Let $s$ be a successful CAS step of a refresh in part\ref{ap-rotate} of Definition~\ref{ap-definition}.
Let $\sigma_L$ be all operations in $Ops(C, x_L)$ that do not already have an arrival point at $x$ and $\sigma_R$ be all operations in $Ops(C, x_R)$ that do not have an arrival point at $x$.
In this case, $s$ is the arrival point of all operations in $\sigma_L$ and $\sigma_R$ at $x$, such that $Ops(C', x)$ = $Ops(C, x) \cdot \sigma_L \cdot \sigma_R$.

In $C$, Let $v_L$ and $v_R$ be the Versions stored in $x_L.version$ and $x_R.version$, respectively. These two version trees represent sets $\kappa_L$ and $\kappa_R$. 
By \obsref{rd-chd-ver}, $v_L$ was read at some configuration in or before $C$ when $x_L$ was reachable from $x$. 
Consequently, the set $\kappa_L$ consists of keys obtained from sequentially executing operations for the keys in $x_L$'s previous version and for the new keys in $v_L$. Same applies for $v_R$ and $\kappa_R$. 


When $s$ changes $x.version$ to a new Version $v$ with children $v_L$ and $v_R$., then in $C'$, $v$ represents a set of keys $\kappa$ $=$ $\kappa_L$ $\cup$ $\kappa_R$. This is the set of keys which will result in sequentially executing operation in $Ops(C', x)$.
Additionally, since the step $s$ does not change responses of operations that arrived from its children the claim 2 is preserved.



\end{proof}


\begin{corollarynonum}
Each update that terminates returns a response consistent with the linearization ordering.    
\end{corollarynonum}
\begin{proof}
    By Definition~\ref{ap-definition}, for an insert or delete that returns false, the response value associated with it is \texttt{false}. Similarly, for an insert or delete that returns true, the response value associated with it is \texttt{true}.
    In both cases, the responses are carried up along with their corresponding operations when they arrive at root. 
    From Invariant~\ref{inv:version_keys_resp}, it follows that responses of operations are consistent with performing the operations sequentially in their linearization order.    
\end{proof}


The following invariant follows from the way \op{Refresh} works.
Essentially, all the Version nodes are immutable and they satisfy the Invariant at the time they are created at Line ~\ref{lin:allocateVersion} in \op{Refresh}. 
\begin{invariant}
    For every Version v that has children, v.size = v.left.size + v.right.size
\end{invariant}

\begin{corollary}
For every Version v, v.size is the number of leaves in the tree rooted at v that contains key from the universe of keys.
\label{cor:correctsize}
\end{corollary}

\begin{lemma}
    The result returned by each query operation is consistent with the linearization.
\end{lemma}
\begin{proof}

Queries are defined to be linearlized at the moment they read $root.version$. This provides a query with an immutable snapshot of the Version tree. When Invariant~\ref{inv:version_keys_resp} is applied at root, the Version tree obtained by the query is a BST and accurately reflects the precise set of keys that would result from sequentially performing all update operations in the $Ops$ sequence of the root.
Additionally, Corollary~\ref{cor:correctsize} ensures that the number of leaves at a Version tree rooted at any node $v$ taken from the snapshot is same as the value of $size$ in $v$.

\end{proof}